\documentclass[a4paper,11pt]{article}
\usepackage[authoryear]{natbib}
\bibpunct{(}{)}{;}{a}{,}{,}
\usepackage{array}
\newcolumntype{M}{>{\centering\arraybackslash}m{\dimexpr.25\linewidth-2\tabcolsep}}

\usepackage{graphicx}
\usepackage{amsmath}
\usepackage{amssymb}
\usepackage{amsthm}
\usepackage{xcolor}
\usepackage{hyperref}

\hypersetup{
    colorlinks,
    linkcolor={red!50!black},
    citecolor={blue!50!black},
    urlcolor={blue!80!black}
}

\newtheorem{theorem}{Theorem}[section]

\newtheorem{definition}{Definition}[section]
\newtheorem{lemma}{Lemma}[section]
\newtheorem{corollary}{Corollary}[section]
\newtheorem{claim}{Claim}[section]

\renewcommand{\setminus}{-}

\renewcommand{\thefootnote}{\fnsymbol{footnote}}

\vspace{-1cm}
\title{The Effect of a Finite Time Horizon in the Durable Good Monopoly Problem with Atomic Consumers \footnote{We give thanks for relevant feedback provided by Juan Ortner, Anton Ovchinnikov, Sven Feldmann, Jun Xiao, and to seminar participants at London School of Economics, Sauder School of Business, Fuqua School of Business, University of Auckland Business School, University of New South Wales, Universit\'{e} Libre de Bruxelles, Carleton University, NetEcon conference 2013 and the Web and Internet Economics conference 2014.}\\[-1cm]}
\begin{document}
\bibliographystyle{plainnat}

\maketitle

\renewcommand*{\thefootnote}{\fnsymbol{footnote}}

\vspace{-1cm}
\begin{center}
{\sc Gerardo Berbeglia}\footnote{Melbourne Business School, University of Melbourne. Email: {\tt g.berbeglia@mbs.edu} },
{\sc Peter Sloan}\footnote{McGill University. Email: {\tt ptrsln@gmail.com} }
and
{\sc Adrian Vetta}\footnote{McGill University. Email: {\tt }adrian.vetta@mcgill.ca}
\end{center}

\renewcommand*{\thefootnote}{\arabic{footnote}}
\addtocounter{footnote}{-3}

\
\noindent\textbf{Abstract.}
A durable good is a long-lasting good that can be consumed repeatedly over time, and a
duropolist is a monopolist in the market of a durable good. In 1972, Ronald Coase conjectured that a duropolist who lacks commitment power cannot sell the good above the competitive price
if the time between periods approaches zero. Coase's counterintuitive conjecture was later proven by \citet{gul1986foundations} under an infinite time horizon model with non-atomic consumers.
Remarkably, the situation changes dramatically for atomic consumers and an infinite time horizon. \citet{bagnoli1989durable} showed the existence of a subgame-perfect Nash equilibrium
where the duropolist extracts all the consumer surplus. 
Observe that, in these cases, duropoly profits are either arbitrarily smaller or arbitrarily larger than
the corresponding static monopoly profits -- the profit a monopolist for an equivalent consumable
good could generate. In this paper we show that the result of \citet{bagnoli1989durable}
is in fact driven by the infinite time horizon. Indeed, we prove that for finite time horizons and atomic agents, in any equilibrium satisfying the standard skimming property, duropoly profits are at most an additive factor more than static monopoly profits.
Specifically, this additive factor is equal to the static monopoly price, which implies that duropoly profits never exceed twice the static monopoly profits.
%In particular, duropoly profits are always at least static monopoly profits but never exceed twice the static monopoly profits.

Finally we show that, for atomic consumers, equilibria may exist that do not satisfy the
skimming property. We prove that amongst all equilibria that maximize duropoly profits,
at least one of them satisfies the skimming property in the two-period setting. %We conjecture that this is true for any number of time periods.\\ \\
\noindent\emph{Keywords}: durable goods monopoly, discrete buyers, profit bounds, inter-temporal
price discrimination, skimming property.
%\end{abstract}

\renewcommand*{\thefootnote}{\fnsymbol{footnote}}

\renewcommand*{\thefootnote}{\arabic{footnote}}

\newpage

\section{Introduction}

A {\em durable good} is a long-lasting good that can be consumed repeatedly over time.
Theoretically less is known about durable goods than their more well-studied
counterparts, consumable and perishable goods. However, on the practical side, durable goods abound
and are very familiar to us. For example, many of the most important consumer items
are (at least to some extent) durable, such as land, housing, cars, etc.
A {\em duropolist} is a monopolist in the market of a durable good --
topically, duropolists include several well-known purveyors of digital goods. Indeed, Amazon has recently been awarded a US patent (8,364,595) for establishing market places for second-hand digital-content items, and Apple has recently applied for a similar patent (20130060616).

Pricing a durable good is not as simple as it may appear.
Specifically, whilst a durable good is more valuable to the consumer than an equivalent perishable good, it is questionable
whether a duropolist has additional monopoly power beyond that
of an equivalent monopolist for a perishable good.
Indeed, quite the opposite may be true.
In 1972, Nobel recipient Ronald Coase made the startling conjecture that, in fact,
a duropolist has no monopoly power at all!
Specifically, a duropolist who lacks commitment power cannot sell the good above the competitive
price if the time between periods approaches zero \citep{coase1972durability}. The intuition
behind the Coase conjecture is that if the monopolist charges a high price then consumers anticipate a future
price reduction (as they expect the duropolist to later target lower value consumers) and therefore they
prefer to wait. The duropolist, anticipating this consumer behaviour, will then drop prices down to the competitive level.
In essence, the argument is that a duropolist is not a monopolist at all:
the firm does face stiff competition -- not from other firms but, rather, from
future incarnations of itself! This is known as the {\em commitment problem}: the duropolist cannot credibly commit to charging a high price.

%There have since been several proofs and disproofs
%of the conjecture for under assorted economic models and time horizons. We discuss this, and also
%real-world strategies that duropolists use to avoid the conundrum highlighted by Coase.
%Our main results are to quantify how well various price mechanisms perform.
%Specifically, we give tight bounds, for the finite time horizon case, on the
%relative profitabilities of these mechanisms in terms of the number of time periods
%and the total number of consumers. In doing so, we quantify the
%the extent to which a duropolist can generate higher profits than an equivalent
%monopolist for a perishable good.

The Coase conjecture was first proven by \citet{gul1986foundations}
under an infinite time horizon model with non-atomic consumers.
They showed that if buyers strategies are stationary (that is, the
distribution of consumers after the duropolist announces a price $p$, lower than all previous prices,
is independent of the prior price history) then, as period length goes to zero, the duropolist's first price offer
converges to the lowest consumer valuation or the marginal cost, whichever is higher.
\citet{ausubel1989reputation} later showed that if the stationary condition is relaxed, the
duropolists profits at subgame perfect equilibria can range from Coasian profits to the static monopoly
profit.\footnote{However, profits larger than the Coasian value occur only in the case where there is no gap between the lowest consumer value and the marginal cost of production (e.g. for $c=0$, the lowest consumer value is 0).} \citet{stokey1979intertemporal}
studied pricing mechanisms
for duropolists that {\em possess} commitment power in a continuous time model. She showed that duropolists
can then attain the static monopoly profit by committing to a fixed price; all sales are then made at the beginning of the game. \citet{mcafee2008capacity} examined the Coase conjecture in a model where there is small cost for production capacity which can be augmented at each period. In this setting, the authors showed that in the the monopoly profits are equal to those that can be obtained if she could commit ex ante to a fixed capacity. Recently, \citet{ortner2016durable} studied a model where the duropolist incurs in a stochastic cost.

%% have we extended stokey 79 result to the case of finite periods and finite players? is that result new?

Underlying the results of \citet{gul1986foundations} is the assumption of an
infinite time horizon. There are nevertheless situations where trade must take place before a hard deadline.
For example, consider a TV network selling advertising space a week in advance of a show. Theoretical and empirical
evidence of the strong effects of deadlines have been observed in many bargaining contexts such as in contract
negotiations and civil case settlements -- see, for example, \citet{cramton2002strikes}, \citet{williams1983legal}
and \citet{fuchs2011bargaining}.

If there is a finite horizon and non-atomic consumers, a feasible action
for the duropolist is to decline to sell goods until the final period and then announce the static monopoly price,
obtaining the static monopoly profits discounted to the beginning of the game. Although this strategy is
not an equilibrium, \citet{guth1998durable} showed that when consumer valuations follow a uniform distribution,
there exists a subgame perfect equilibrium, as period lengths approach zero, in which the duropolist profits converge to the static monopoly profits
discounted to the beginning of the game.
%%guth et al say it is the unique
%equilibrium satisfying their eq refinement. can we translate that uniqueness result into the context of the other papers?

% note: In Guth,  page 223, prop1: if the discount rates R and r are equal to 1, why the monopolist revenue is
%equal to something less than 1/4 in the limit? Isn't always possible
% in this case for the monopolist to announce the SMP at the end of the game to get 1/4?

All these results assumed non-atomic consumers. Less is known about the duropolist problem with
atomic consumers. The major difficulty in analyzing equilibria in this setting is that a deviation
from a single consumer can modify the equilibrium price path. This characteristic, as we will later show, permits the
existence of subgame perfect Nash equilibria where the {\em skimming property} does {\bf not} hold, i.e. an equilibrium
where a buyer with a higher valuation might buy later (and at a lower price) than a buyer with lower
valuation.
\citet{bagnoli1989durable} studied the duropoly problem with atomic consumers.
When the time horizon is infinite, they proved another surprising result: the
existence of a subgame perfect Nash equilibrium in which the duropolist extracts all the consumer surplus. To
obtain this, they considered the following pair of strategies. The duropolist strategy, dubbed \emph{Pacman}, is
to announce at each time period, a price equal to the valuation of the consumer with the highest value who has yet to buy. The
strategy of each consumer, dubbed \emph{get-it-while-you-can}, is to buy the first time it induces a non-negative utility.
If there is an infinite horizon, and the discount factor is large enough, then these strategies are
sequential best responses to each other. This equilibrium refutes the Coase conjecture when there is a finite set of atomic
buyers and an infinite time horizon. Indeed, it suggests that a duropolist may have perfect price discriminatory power!
Moreover, it shows there exist subgame perfect Nash equilibria where duropoly profits exceeds the static monopoly profits
by an unbounded
factor.\footnote{For example, consider a game with $N$ consumers where buyer $i$ has a valuation of $1/i$. Then, as $N$ gets large, duropoly profits under the Pacman strategy approach $\log N$ whilst
static monopoly profits are clearly equal to $1$.} Later, \citet{von1995coase} showed that under certain conditions Pacman is the only equilibrium.

Another setting, in which the extremely large profits of the Pacman equilibrium of \citet{bagnoli1989durable}
cannot exist, was proposed by \citet{cason2001durable}. Instead of assuming a duropolist with perfect
information, the authors constructed a two-buyer and two-valuation model with infinite time periods in which the duropolist does not know exactly whether a consumer is of high type or of low type. They showed that in these games there exists a unique equilibrium that is Coasian.

Recently, \citet{montez2013inefficient} studied the duropoly problem under infinite horizon with atomic consumers that have two-types (high value consumers and low value consumers) and exactly two consumers are of the high-value type. He showed that there are sometimes inefficient equilibria where the time at which the market clears does not converge to zero as the length of the trading periods approaches zero.

\citet{bagnoli1989durable} also examined very small examples with a finite time horizon case.
They showed that in these games of two or three consumers, it is possible to
obtain subgame perfect equilibria where the duropolist extracts more revenue than the
price commitment strategy of \citet{stokey1979intertemporal}.
These examples again refute the Coase conjecture but they also suggest that,
even for finite time horizons, the duropolist may have {\bf more} monopoly power than the
equivalent static monopolist.
This is very interesting because, whilst duropolists are not believed to be powerless in practice,
the standard assumption is that duropolists are weaker than monopolists for consumable goods.
Indeed this argument has been accepted by the Federal Courts in the United States;
a history of duropolies and the law can be found in \citet{orbach2004durapolist}.

\subsection{Contributions}
Several questions arise from the work of \citet{bagnoli1989durable}.
Does this phenomenon (of duropoly profits exceeding static monopoly profits) arise for more natural
games where the number of consumers and the number of time periods is much larger than three?
If so, how can we characterize the monopolist profit maximizing strategy?
Finally, can we quantify {\em exactly} how much more profit a
duropolist can obtain at equilibria in comparison to a static monopolist?

Our main contributions is to answer those three questions. To achieve this, we first
characterize, in Section \ref{sec:opt-cond}, a class of subgame perfect equilibria that satisfy the standard skimming property:
high-value consumers buy before lower-valued consumers.\footnote{\citet{bagnoli1989durable}
state that such a characterization would be extremely interesting.} %We do this both in the complete information setting as well as a setting where market participants have limited information about other consumer valuations.
Our main result, proven in Section \ref{sec:results}, is then that, at equilibria,
duropoly profits are {\em approximately} equal to the static monopoly profits. Specifically we show that duropoly profits are at most the static monopoly profits plus the static monopoly price regardless of the number of consumers, their values, and the number of time periods.
It follows from this result that duropoly profits are at most twice the static monopoly profits. We also prove that this factor two bound is tight: we construct a (infinite) family of examples where duropoly profits approach double the static
monopoly profit as the number of consumers goes to infinity. This construction also gives tightness for the additive bound
as the difference between the duropoly profits and the static monopoly profits tends to the static monopoly price.

We believe that our main result sheds light into this classical problem in at least four ways.
To begin, this is the first theoretical result that concurs with the practical experience
that duropolists and static monopolists have comparable profitability.
(Recall that previous theoretical works have suggested that the duropolist either has no monopoly power
or has perfect price discriminatory power).

Second, the result that a duropolist can do up to an additive amount
better using a threat-based strategy rather than a price-commitment strategy
is actually best viewed from the opposite direction.
Specifically, a duropolist can obtain almost the optimum profit (losing at most an additive amount equal to
the static monopoly price) by mimicking a static monopolist via a price commitment strategy.
From a practical perspective this is important because a price-commitment strategy can generally
be implemented by the duropolist very easily, even with limited consumer information.
Furthermore, price-commitment strategies can be popular with consumers
as they are typically introduced within a money back guarantee or
envy-free pricing framework. In contrast, a threat based optimization
strategy is harder to implement and can antagonize consumers.

Thirdly, the standard view in the literature is that the surprising and well-known result
of \citet{bagnoli1989durable}, namely that the duropolist can extract all consumer
surplus, is due to the assumption of atomic consumers. Our results show that this is not
true in general -- their result is driven by the infinite time horizon. For finite time horizons, the power of a duropolist
is limited. This is true even when the Pacman strategy is an equilibrium; indeed, we show that the
\emph{Pacman} strategy can be an equilibrium in finite time horizon games only under very specific
conditions\footnote{The \emph{Pacman} strategy will not produce an equilibrium in games with finite horizons {\em unless} consumer valuations decrease exponentially.} -- see Section \ref{sec:ConSur}.

Finally, the main result highlights a distinction in how the time horizon affects bargaining power.
With non-atomic consumers, a finite time horizon increases the bargaining power of the duropolist. In \citet{guth1998durable},
a finite time-horizon increases duropolist profits from the Coasian result to the static monopoly profits.
With atomic consumers, the finiteness of the time horizon,
reduces duropolist profits from all consumer surplus to approximately the static monopoly profits (and at most twice that).

To conclude the paper, we examine (in Section \ref{sec:NoSkimSum}) subgame perfect equilibria in the absence of the skimming-property. We provide the first example of a subgame perfect equilibrium in which the skimming-property does not hold. Furthermore, we prove that for two-period games, amongst all equilibria that maximize duropoly profits at least one satisfies the skimming property.

% write about bagnoli et al, Cason and Sharma 2001, reynolds 2000

%\citet{guth1998durable} has studied
%\citet{stokey1981rational} has

\section{The Model}\label{sec:model}
We now present the durable good monopoly model of \citet{bagnoli1989durable} that we will analyze in the subsequent sections.
Consider a durable good market with one seller (a duropolist), $N$~consumers and a finite horizon of $T$ time periods.
The $N$ consumers have valuations $v_1 \ge v_2 \ge \dots \ge v_N$\footnote{We also use notation $v(y_j)$ instead of $v_{y_j}$ in certain cases to avoid nested subscripts} and the firm can produce units of the good at a
unitary cost of $c$ dollars. Here we assume, without loss of generality, that $c=0$.
Consequently, profit and revenue are interchangeable in this setting.

We can view this as a sequential game over $T$ periods.
At time $t$, $1 \le t \le T$, the firm will select a
price $\mu_t$ to charge for the good.\footnote{Note neither the model of \citet{bagnoli1989durable}, that we analyze in this paper, nor any of the other durable good monopoly models mentioned in the literature review section allow for discriminatory pricing mechanisms in which two or more consumers can be charged different prices in the same time period.}
The duropolist seeks a pricing strategy that maximises her revenue, namely $\sum_{t=1}^{T}(x_t \cdot \mu_t)$,
where $x_t$ denotes  the number of consumers who buy in period~$t$.
%We will restrict to monopolist pricing strategies that are dependent on the set of consumers left and the current time period, i.e. $\mu: [N] \times [T] \rightarrow R$.

Each consumer $i$ desires at most one item and seeks to maximize her {\em utility},
which is $v_i - \mu_t$ if she buys the good in period $t$.\footnote{Discount factors
can easily be introduced into the model.} The consumers decide simultaneously if they will buy an item for $\mu_t$. %The duropolist observes the total number of transactions that occur at time $t$, however, the consumer identities of those who bought remain undisclosed.
The game then proceeds to period $t+1$. If a consumer doesn't buy an item before the end of period $T$ her utility is zero.

For such a sequential game, the solutions we examine are pure subgame perfect Nash equilibria that satisfy the standard \emph{skimming property} defined below.

\begin{definition}[Skimming property]
An equilibrium satisfies the \emph{skimming property} if whenever a buyer with value $v$ is willing to buy at price $\mu_t$, given the previous history of prices $\mu_1,\mu_2,\hdots,\mu_{t-1}$, then a buyer with value $w>v$ is also willing to buy at this price given the same history.
\end{definition}

For subgame perfect Nash equilibria (SPNE) that satisfy the skimming property, consumers' strategies can be characterized using a cutoff function. Given a history of prices $h_t=(\mu_1,\mu_2,\hdots,\mu_{t-1})$ and the current offered price $\mu_t$, consumers with valuations above cutoff $\kappa(h_t,\mu_t,t)$ buy and consumers with valuations below the cutoff do not buy  (see \citet{tirole} for a discussion).  When consumers are non-atomic it can be shown that all subgame perfect equilibria satisfy the skimming property \citep{fudenberg1985}. In the case of atomic consumers and an infinite time horizon, the monopolist can extract all consumer surplus using the \emph{pacman} strategy, in which case the skimming property is clearly satisfied. Intuitively, the skimming property says that higher value consumers pay a higher (or at least equal) price compared to consumers with a lower valuation. %In real markets, this phenomenon is widely observed.\footnote{In experimental economics this is tested via screening techniques.}

\citet{ausubel1989reputation} define two special types of SPNE that are \textit{Markovian} in the sense that they depend only on the most recent information available. A SPNE is a \textit{weak-Markov} equilibrium if consumers' accept/reject decisions depend only on the current price and period. A SPNE is a \textit{strong-Markov} equilibrium if, in addition to the weak Markov property, the duropolist conditions her strategy only on the payoff-relevant part of the history. In the infinite horizon case, this is the set of remaining consumers. In the finite horizon case, it may depend on the number of periods left as well. When consumers are non-atomic, such equilibria only exist when the Weak-Markov property is satisfied and $= \kappa(\mu_t,t):= \kappa(h_t,\mu_t,t) $ (since consumers' accept/reject decisions  do not depend on the price history) is strictly increasing in $\mu_t$ \citep{fudenberg1985}. In the atomic case, we can obtain strong-Markov equilibria even if $\kappa$ is constant over an interval.

When constructing an SPNE in the atomic finite-horizon model, we will therefore restrict the strategy space of the duropolist and consumers so that they satisfy the strong-Markov conditions: the prices the duropolist chooses are a function that depends only on the remaining consumers and the number of time periods left, that is, $\mu:\mathcal{P}([N])\times T \rightarrow \mathcal{R}^+$ and the consumers strategies are such that $i$ buys in period $t$ iff $v_i \geq \kappa(\mu_t,t)$ for some function $\kappa$. However, our main result from Section $\ref{sec:results}$ only requires the equilibrium satisfies the skimming property.

%WHAT SHOULD WE SAY WITH REGARS OUR MODEL? IT IS COMPLETE INFORMATION, BUT WE LATER SHOW THAT IT EXTENDS TO AN INCOMPLETE SETTING?

%In this paper we will characterize the sub-game equilibria (Section \ref{sec:opt-cond}) and provide an upper bound on the duropolist profits (Section \ref{sec:results}) for exactly those equilibria that satisfy the standard skimming property. In Section \ref{sec:no-skim} we show that, surprisingly, there exist sub-game perfect equilibria that do not satisfy the skimming property. We prove that the upper bound still holds for non-skimming equilibria whenever $T=2$ and conjecture that it holds for all $T$.

Observe that in the model of \citet{bagnoli1989durable} which we studied in this paper, the values of each consumer are known to all participants in the market. In Appendix 2 we extend the equilibrium results to a restricted incomplete information setting, where the duropolist knows the distribution of values, but not which value corresponds to which consumer. %For example, she may know the consumer values are $\{100,50,30\}$, but does not know which of these three is the value of Consumer 1. Similarly, each consumer knows the set of values and which is their own value, but does not know which values correspond to the other consumers. Note that, in this model,
% the $N$ consumers may have $N$ very different types; we do not restrict consumers to be drawn from a common known distribution.

\subsection{An Example}\label{sec:example}
We now present an small example to illustrate the model and the concepts involved.
Consider a two-period game with the following consumers' valuations.

\begin{table}[h]
\centering
\caption{Example of a game with 4 consumers.}
\begin{tabular}{|c|c|}
\hline
\textbf{Consumer} & \textbf{Consumer value}  \\
\hline
1 & 100 \\
\hline
2 & 85 \\
\hline
3 & 80 \\
\hline
4 & 50 \\
\hline
\end{tabular}
\label{ExValue}
\end{table}

\begin{table}\label{table_threat_prices}
\centering
\caption{Threat prices}
\begin{tabular}{|c|c|c|}
\hline
\textbf{Consumer} & \textbf{Consumer value}  & \textbf{Threat price} \\
\hline
1 & 100 & 80\\
\hline
2 & 85 & 80\\
\hline
3 & 80 & 50\\
\hline
4 & 50 & 50\\
\hline
\end{tabular}
\label{ExThreat}
\end{table}

%\begin{table}[h]
%\table{Consumer valuations}{
%       \begin{tabular}{{|M|M|}}
%\hline
%               Consumer & Consumer Value \\
%\hline
%              1& 100 \\
%\hline
%          2&     85 \\
%\hline
%             3 & 80 \\
%\hline
%        4 &       50 \\
%\hline
%       \end{tabular}}
%\label{ExValue}
%\end{table}

Denote by $\Pi^D$ and $\Pi^M$ the revenue obtainable by the duropolist  and the
corresponding static monopolist, respectively.
Then the static monopoly profit $\Pi^M$ is equal to $240$, obtained by selling to the top three consumers for a price of $80$. However, the duropolist can, in fact, extract a revenue of $260$. Furthermore, the corresponding equilibrium
satisfies the skimming property: no consumer will buy earlier than another consumer with a higher valuation.
To understand SPNEs in this game, let's begin with a subgame
comprising of only the final (second) time period. In such a subgame, it is a dominant strategy for all consumers who have not yet bought
to pay any price less than or equal to their value. Consequently, in the final period, it is a dominant strategy for the duropolist to charge the
static monopoly price \textit{as calculated with respect to the set of consumers who have not yet bought}. Note that these strategies satisfy the strong-Markov property as everyone remaining with value above the price will buy, everyone else will not buy, and the price depends only on the set of consumers remaining.

Now consider the first time period. If the skimming property is satisfied, then there will be a cut-off point $j_1$ at
which consumers $j\leq j_1$ buy and consumers $j>j_1$
wait until period 2. In order for this to be an equilibrium, the consumer $j_1$ must prefer buying in period 1 to period 2.
Therefore, the duropolist can charge no more than the static monopoly price as calculated if all consumers $j \geq j_1$
wait until period 2. We call this price the \emph{threat price} for consumer $j_1$. The threat prices are listed in Table \ref{table_threat_prices}.

%
%\begin{table}
%\tbl{Threat Prices}{
%       \begin{tabular}{{|M|M|M|}}
%\hline
%               Consumer & Consumer Value & Threat Price\\
%\hline
%              1& 100 & 80 \\
%\hline
%          2&     85 & 80\\
%\hline
%             3 & 80 & 50 \\
%\hline
%        4 &       50 & 50 \\
%\hline
%       \end{tabular}}
%\label{ExThreat}
%\end{table}

The consumers' strategies then correspond to ``buy in period 1 if and only if $\mu_1$ is at most their threat price", whilst the
duropolist's strategy is to charge the threat price which maximizes the total revenue. The period 2 strategies are the
dominant strategies described above: remaining consumers pay up to their value, while the duropolist charges
the static monopoly price calculated for the set of consumers that are left.

Charging $\mu_1=80$ means that the top two consumers would buy in period 1, while the last two consumers would wait
until period 2 and buy at $\mu_2=50$ (the static monopoly price for the remaining two consumers is $50$). The total
profit would therefore be $260$. It is easy to see that charging $50< \mu_1 <80$ gives a smaller profit.
Moreover, if $\mu_1>80$ then no consumers will buy in the first period. This would lead to a profit of $240$
as the static monopoly price would then be charged in the final period.
Finally charging $\mu_1=50$ would result in all consumers buying in period 1, as the duropolist
is guaranteed to charge at least $50$ in period 2. The total profit would then be $200$. Thus, for a profit maximizing
duropolist we have $\Pi^D=260$.
So duropoly profits are greater than static monopoly profits. For additional comparisons,
Coasian profits are $\Pi^C=200$ since the competitive price is $50$, and Price Discriminatory
profits are $\Pi^{PD}=315$, that is, the consumer surplus.

Observe that these strategies satisfy the skimming property, as threat prices are monotonically increasing with consumer
value, and satisfy the weak-Markov property, as $\kappa(\mu_t,t)$ is the smallest consumer value such that his threat price is larger than $\mu_t$. Note that $\kappa(\mu_1,1)$ is constant for all $\mu_1 \in (50,80]$.  Since the duropolist's strategy depends only on the values of remaining players in each period (trivially in period 1) it also satisfies the strong-Markov property. Furthermore, it is easy to see that this is an equilibrium. Consumers 3 and 4 would be worse off buying in the first period,
while both consumers 1 and 2 would still pay 80 if either deviated by waiting until the last period. Similarly, the duropolist
would find no buyers if he charged more than 80 in period 1, turning the game into a one-period static game and
giving him revenue $\Pi^M$, while charging anything between 80 and 50 would yield the same sales schedule
but with lower profit. Charging 50 or lower would result in all consumers buying in period 1 for profit at most 200.

We remark that the equilibrium property for the consumers arises from a simple property of static monopoly prices.
Take the set of consumers $\{j\geq j_1\}$ and compute the static monopoly price for these consumers. This is $j_1$'s threat
price. Now if we take $j_1$ and replace him by a higher valued consumer then the static monopoly price can only rise.
Hence, as long as $\mu_1$ is $j_1$'s threat price, any consumer of higher value that refuses to buy in period 1 would
be charged at least $\mu_1$ in period 2, ensuring it is a best response to buy in period 1.

In Section \ref{sec:opt-cond}, we show how these arguments can be extended to give subgame perfect equilibria
conditions for games with more than two time periods and explain how, given these constraints, a duropolist
can efficiently maximize profits.

\section{Sub-Game Perfect Equilibria Conditions}\label{sec:opt-cond}

%\begin{eqnarray}
%R^*(y,T) &=& \max_{a \in [0,n-y]} a \cdot P(y,T) \label{rev_equation_no_trust_sorted_one_period}\\
%P^*(y,T)&=& \arg \max_{{V}_{y+a}: a \in [0,n-y]} a\cdot {V}_{y+a} \label{price_equation_no_trust_sorted_one_period} \\
%X^*(y,T) &=& \arg \max_{a \in [0,n-y]} a\cdot {V}_{y+a} \label{x_equation_no_trust_sorted_one_period}\\
%R^*(y,t) &=& \max_{a \in [0,n-y]} \left[ a \cdot  P(y+a-1,t+1) +  R(y+a,t+1) \right] \label{rev_equation_no_trust_sorted}\\
%P^*(y,t)&=& P^*(y+X^*(y,t)-1,t+1)\label{price_equation_no_trust_sorted} \\
%X^*(y,t) &=& \arg \max_a R^*(y,t) \label{x_equation_no_trust_sorted}
%\end{eqnarray}
%
%$X^*(y,t),$and$P^*(y,t)$ represent the number of items sold and the price at time period $t$ given that the top $y$ players already have bought the item prior to time period $t$. The value $R^*(y,t)$ represents the monopolist's revenue from time period $t$ to $T$ (discounted to time period $t$) assuming the top $y$ players have bought prior to period $t$. The total revenue for the monopolist can be computed as  $R^*(0,0)$.
%\begin{theorem}
%The optimal sequence of prices announced by the duropolist is $(P(0,0), P(X(0,0),1), P(X(0,0)+X(X(0,0),1)), \hdots,..)$ that satisfy the skimming property, the duropolist price at period $t$, given that $y$ consumers have already bought is equal to $P^*(y,t)$.
%\end{theorem}
We now characterize the subgame perfect equilibria that satisfy the strong-Markov conditions and maximize duropolist profits. To do this we reason backwards from the final time period $T$. It is easy to
determine the behaviour of rational consumers  and a profit maximizing duropolist at time $T$. Given this information, we can determine the behaviour of rational consumers at time $T-1$, etc.

%We begin with such characterization when there is complete information, that is when all participants know which player has each value and the duropolist can observe who buys in each period. Afterwards, we introduce an incomplete information setting and show that the same SGPNEs characterization applies.

%\subsection{Complete Information} \label{section_complete_information}

To formalize this, let $\mathcal{G}_i$ denote the market
consisting of consumers $\{i, i+1, \dots, N\}$, and let
$\Pi(i,t)$ denote the maximum profit obtainable in the market $\mathcal{G}_i$ if we begin in time period $t$.
Thus $\Pi^D=\Pi(1,1)$. Now set $\Pi(i, T+1)= 0$ for all consumers $i$.
Let $p(i, t)$ be the profit maximizing price at period $t$ in the market $\mathcal{G}_i$ beginning at time $t$. First, consider the last period, $T$.
Any consumer $i$ (who has not yet bought the good) will buy in period $T$ if and only if this final price is at most $v_i$.
Therefore, in the market $\mathcal{G}_i$, starting at time $T$, a profit maximizing duropolist will
simply set $p(i,T)$ to be the static monopoly price $p_i$ for the market $\mathcal{G}_i$: $p(i,T) = p_i \equiv v_{j^*(i,T)}$, where
$$ j^*(i, T) = \arg \max_{j\ge i} \,(j-i+1)\cdot v_{j}.$$
Thus, $j^*(i, T)$ denotes the consumer with the lowest valuation who buys in the market $\mathcal{G}_i$ beginning at period $T$. The profit is then
\begin{eqnarray*}
 \Pi(i,T) &=& (j^*(i,T)-i+1)\cdot v_{j^*(i,T)}
 \end{eqnarray*}

In general we will denote by $j^*(i, t)$, the consumer with the lowest valuation who buys (under our proposed strategy) at period $t$ in the market $\mathcal{G}_i$ beginning at period $t$.  Now, suppose we are at time period $T-1$ in the market $\mathcal{G}_i$. If the duropolist at period $T-1$ wishes to sell to consumers $\{i,i+1,\hdots,k\}$, then the announced price has to be at most $k$'s threat price, $p(k,T) \equiv v_{j^*(k, T)}$.
To see this, suppose that the price announced at $T-1$ is higher and the duropolist still expects to sell the item to consumers $\{i,i+1,\hdots,k\}$. Then, if consumer $k$ refuses to buy while all consumers above her buy, the duropolist would, in the final time period $T$ be in the market $\mathcal{G}_k$, and announce a price $p(k,T)$, meaning that consumer $k$ would have benefited from deviating. So, the optimal strategy for the duropolist would be to sell to $k-i+1$ consumers at period $T-1$ at price $v_{j^*(k,T)}$, choosing the value of $k$ such that the profits from periods $T-1$ and $T$ are maximized:
 \begin{align}
j^*(i,T-1) &= \arg \max_{k \geq i} \{ (k-i+1) \cdot p(k, T) + \Pi(k+1,T)\} \nonumber \\
\Pi(i,T-1) &=  (j^*(i,T-1) -i+1) \cdot p(j^*(i,T-1), T) + \Pi(j^*(i,T-1)+1,T) \nonumber
 \end{align}
The price announced at period $T-1$ can then be written as
$$p(i,T-1) = p(j^*(i,T-1),T)$$
Observe then, that in the final period, we will be in the subgame composed of consumers $\{j^*(i,T-1)+1,\hdots,N\}$.

Iterating this argument backwards in terms of the periods, we have that
 \begin{align}\label{eq:characterization_prices}
&\Pi(i,t)& = \ \ & (j^*(i,t)-i+1)\cdot p(j^*(i, t) , t+1) + \Pi(j^*(i, t)+1, t+1) \nonumber \\
&j^*(i, t) &= \ \ & \arg \max_{j\ge i} \,\left((j-i+1)\cdot p(j, t+1)+ \Pi(j+1, t+1)\right)  \\
&p(i, t)& = \ \ & p(j^*(i, t), t+1) \nonumber
 \end{align}

We can generalize the concept of the threat price from our two-period example using the above recursion.
Specifically, we say that the \emph{threat price} $\tau(i,t)$ for consumer $i$ at period $t<T$ under the recursive scheme given in (\ref{eq:characterization_prices}) is the price $i$ is offered in the market $\mathcal{G}_i$ starting at
period $t+1$, namely $\tau(i,t) := p(i,t+1)$. That is, the price offered if $i$ and all consumers of lower value do not buy in period $t$.

We can now define the strategy of the duropolist and the consumers in any subgame. Consider a subgame whose remaining consumers are the set $S$ and let there be $T-t+1$ periods remaining (i.e. we are starting in period $t$). Then by re-indexing the consumer names, the duropolist can treat the subgame as a full game $\mathcal{G'}$ with $T-t+1$ total periods and $S$ as the set of all consumers. She then calculates, for all $i$ and $t$, the prices $p_{\mathcal{G'}}(i,t)$ from the recursion relationship (\ref{eq:characterization_prices}) and chooses the sales schedule which maximizes her profits for $\mathcal{G}'$. She then charges  $\mu_t = p_{\mathcal{G'}}(1,1)$ in period $t$. The consumers buy if and only if the price is less than or equal to their threat price as calculated for $\mathcal{G'}$. By definition, this price only depends on the payoff-relevant part of the history, as it only looks at the consumers remaining in the subgame. Since the price is always equal to the threat price of one consumer $j^*$, we can define $\kappa(\mu_t,t)$ to be the value of this critical consumer, $v_{j^*}$. We can show that this function is monotonically increasing in $\mu$ if the threat prices are decreasing in consumer value (a higher price means a higher valued critical consumer). This is proved in Lemma \ref{lem:threatmono}. Therefore $\kappa$ is indeed a cutoff function for the given consumer strategies. We conclude that, if this strategy profile is an SPNE, it must satisfy the strong-Markov property.

The reader may have noted that our recursion relationship does not allow the duropolist to refuse to sell any items in a period where there are still consumers left who have not yet bought. It can be shown that there is a subgame perfect equilibrium in which a sale occurs in each period until either all consumers have bought the item or the final time period is over. Moreover, this equilibrium achieves at least as much profit for the duropolist as any which allows the duropolist to not sell in some periods. A proof of this is included in Appendix $1$ as Lemma \ref{NeverWait}.

The following series of lemmas establish basic monotonicity results for static monopoly prices, threat prices, and the prices $p(i,t)$ which form the duropolist's equilibrium strategy. The proofs of these lemmas are give in Appendix $1$.

 \begin{lemma}\label{subgame_price_sequence}
%Let the consumers in $\mathcal{G}$ be ordered by value so $v_i \geq v_{i+1}$ for all $i$. Then static monopoly prices are non-increasing in $i$: $p_i \ge p_{i+1}$ for $i=1,\dots, N-1$.
% I deleted the first sentence since p_i and v_1 \geq v_2 etc. was defined right above.
The static monopoly prices on the markets $\mathcal{G}_i$ are non-increasing in $i$: $p_i \ge p_{i+1}$ for $i=1,\dots, N-1$.
\end{lemma}

The following lemma shows that the consumers' strategies defined above satisfy the skimming property.

\begin{lemma}\label{lem:threatmono}
In any game $\mathcal{G}$ with $T$ periods, the threat prices are non-increasing in $i$: for all $i \leq k$ and all $t<T$, $\tau(i,t) \geq \tau(k,t)$.
\end{lemma}

The next two lemmas are required to show that a deviation from a consumer by delaying a purchase or buying early does not yield higher utility.

\begin{lemma}\label{lem:monotone}
Consider two duropoly games, $\mathcal{G}$, with $T$ periods and a set $S$ of consumers, and $\mathcal{G'}$, with $T$ periods and a set $S'$ of consumers such that only the top valued consumer in $S$ and $S'$ differ, and the top valued consumer in $S'$ has the higher value. If we use $p_\mathcal{G} (1,1)$ and  $p_\mathcal{G'} (1,1)$ to denote the first period prices as calculated by the recursion relationship above, then $p_\mathcal{G'} (1,1) \geq p_\mathcal{G} (1,1)$.
\end{lemma}

\begin{lemma}\label{lem:EqPriceDec}
In any game $\mathcal{G}$ with $T$ periods, if the duropolist and consumers follow the strategies described above, then prices are non-increasing in time.
\end{lemma}

We are now ready to prove the following result.

\begin{theorem} The strategies defined above constitute a SPNE. \end{theorem}
\proof Since we can treat any subgame as an instance of a full game with a different set of consumers, there is no loss of generality in assuming that the deviation occurs in the first period of the full game. Let $j^*(1,1)$ be the lowest value consumer sold to in equilibrium and assume that a consumer $x \leq j^*(1,1)$ deviates by not buying in period 1. If $x=j^*(1,1)$, then $x$ is charged her threat price in the next period, which by definition is $p^*(1,1)$, so there is no advantage in a deviation. If $x< j^*(1,1)$, then the remaining consumers for period 2 are $\{x,j^*+1,...,N\}$. We know that if the set of consumers was $\{j^*,j^*+1,...,N\}$, then the price would be $p^*(1,1)$. But by Lemma \ref{lem:monotone}, the price  with consumers $\{x,j^*+1,...,N\}$ must be at least as high as  $p^*(1,1)$. Therefore $x$ cannot gain by delaying her purchase for one period.

One may wonder whether consumer $x$ could benefit from delaying the purchase by more than one period. But this is not the case. Consider the subgame $\mathcal{G}'$ arrived at after $x$ delays purchase for $t-1$ periods in the full game, and after re-indexing so the remaining consumers are sorted by value from highest to lowest. If $x = j^*(1,t)$ for this subgame, the price at period $t+1$ is $p(j^*(1, t),t+1) = p(j^*(j^*(1, t),t+1),t+2)$. This means that the price at $t+2$ will be the same as the price at $t+1$ if consumer $x$ doesn't buy and everyone else follows the equilibrium path. If, on the other hand, consumer $x<  j^*(1,t)$ for subgame $\mathcal{G'}$, the price at period $t+2$ could only increase or stay equal to $p(j^*(j^*(1, t),t+1),t+2)$ by Lemma \ref{lem:monotone}. By repeated use of this argument we conclude that, at equilibrium, no consumer would benefit from delaying its purchase.

It remains to show that no consumer can benefit from buying early.  If a consumer deviates from the equilibrium path by buying early, she pays a price $p^*(1,t)$ when she could have bought in period $t'>t$ at price $p^*(k,t')$ for some $k > 0$. But since prices are non-increasing as a function of time along the proposed sales path (Lemma \ref{lem:EqPriceDec}), she cannot do any better.

It follows that we have a strategy profile which is an equilibrium in every subgame. \qed

To compute such an equilibrium, we can compute $j^*(i,t)$ for each $(i,t)$ going backwards from period $T$, and choose the sales path $x_t$ which maximizes profit. The prices $\mu_t$ are then computed by ``passing back" the next period's threat price. We may solve the corresponding dynamic program to find the maximum profit $\Pi^D$ for the duropolist.

It is easy to see that if there is another strong-Markov SPNE, it cannot result in more revenue for the duropolist. In any proposed SPNE where we sell in every period, any period price cannot be higher than the threat price of any consumer who buys in that period, as otherwise she would earn more profit by waiting one more period. So given a sales schedule, the monopolist can do no better than to charge the threat price of the lowest valued consumer to buy in each period. However, (\ref{eq:characterization_prices}) finds the optimal sales schedule in terms of revenue when the duropolist charges threat prices in each period. Lemma \ref{NeverWait} covers the case where a strong-Markov SPNE may choose not to sell in one or more periods when consumers remain to buy. The full result we have, then, is

\begin{corollary}\label{cor:uniquemax}
The optimal revenue $\Pi^D$ given by the dynamic program derived from Equations \eqref{eq:characterization_prices} is the maximum revenue obtainable by a strong-Markov SPNE.
\end{corollary}

\section{When can the Duropolist Extract all the Consumer Surplus?} \label{sec:ConSur}
As discussed, \cite{bagnoli1989durable} proved that a duropolist who faces atomic consumers with an infinite time horizon can always extract all consumer surplus. They left open the case of finite time horizons. Although such equilibria may still exist under a finite horizon, the conditions required for their existence are very restrictive. Indeed, applying the techniques we have developed, we characterize in this section necessary and sufficient conditions for this phenomenon to happen. We require some definitions and a few lemmas.

Let $\mathcal{G}$ be a game with $N$ consumers with valuations $v_1 \geq v_2, \hdots,\geq v_N$ and let $p_i$ be the static monopoly price of the subgame consisting of consumers $\{i, i+1, \dots, N\}$. More generally, given a subset $S \subseteq [N]$ we define $p(S)$ to be the static monopoly price of the subgame consisting of consumers $S \subseteq [N]$.

\begin{lemma} \label{highest_value_lemma}
A game $\mathcal{G}$ satisfies $p_i=v_i$ for all $i \in [N]$ if and only if $p(S) = \max\{v_x: x \in S\}$ for all $S \subseteq [N]$.
\end{lemma}
\begin{proof}
Suppose there exists a subset $S \subseteq [N]$ such that $p(S) < v_i$ where $i = \arg \max \{v_x: x \in S\}$. Let the valuations of the consumers in $S$ be $v_i \geq v'_2 \geq v'_3 \geq \hdots \geq v'_{|S|}$. Then, we have $j\cdot v'_j > v_i$ for some $j>1$. But then $p_i < v_i$ since by setting a price of $v'_j$ in the subgame with consumers $\{i,i+1,\hdots,N\}$ yields a profit of at least $j \cdot v'_j > v_i$ (since in the market $\mathcal{G}_i$ there may be more consumers with valuations between $v'_j$ and $v_i$). The remaining implication follows directly.
\end{proof}

Let $w_1 > w_2 > \hdots > w_{M}$ denote the $M$ \emph{distinct} consumer valuations sorted in decreasing order. Let $n_i$ denote the number of consumers with value $w_i$. The following technical lemma (whose proof is in the Appendix 1) is required to prove the main result of this section. We set $w_i = n_i = 0$ for all $i > M$.

%We use the following notation: let $w_1 > w_2 > \hdots > w_{M}$ denote the $M$ \emph{distinct} consumer valuations sorted in decreasing order. The proof of the following lemma is deferred to the appendix.

\begin{lemma}\label{lem:pacman1}
If $p_i=v_i$ for all $i \in [N]$, the following inequality holds for every natural number $\beta \geq 2$, $k=2,\hdots,\beta$ and all $i=1,\hdots,k-1$,
$$n_i \cdot w_i - n_i \cdot w_k - n_{\beta+i}w_{\beta+i} \geq 0.$$
\end{lemma}

Recall that the duropolist strategy named \emph{Pacman} is to announce at each time period, a price equal to the valuation of the consumer with the highest value who has yet to buy, and that the consumer strategy known as \emph{get-it-while-you-can} is to buy the first time it induces a non-negative utility. The next lemma gives sufficient conditions for such strategies to be at equilibrium. Observe that we are not imposing that the number of periods is greater than or equal to the number of different consumer valuations. This makes the argument not trivial.

\begin{lemma} \label{pacman_equilibrium_lemma_old_version}
If $p_i=v_i$ for all $i \in [N]$, then there exists an equilibrium in which the duropolist uses the pacman strategy and consumers follow the get-it-while-you-can strategy.
\end{lemma}
\begin{proof}
We proceed by induction in the number of time periods. For games with a single period (i.e., $T=1$) the lemma holds because the duropolist announces the optimal static monopoly price, which is $p_1$.
Suppose now that the lemma holds for all games with at most $T-1$ periods and consider a game with $T$ periods. Let $A$ be the set of consumers that buy at period 1 under an equilibrium $\mathcal{E}$. Observe that by Lemma \ref{highest_value_lemma}, the subgame that begins at period 2, with consumers $[N] \setminus A$ satisfies the that $p_i=v_i$ and therefore there exist an equilibrium where the duropolist uses the Pacman strategy from then on. This means that consumers expect zero profits whenever they don't buy in the first time period, and therefore they would buy in the first time period at any price that is not above their valuation. If $M \leq T$ the duropolist may announce at $t=1$ the price $\mu_1=v_1$. All consumers with a valuation of $v_1$ would buy and, by the inductive hypothesis and the fact that the different valuations in the remaining game is still less than the time periods left ($M-1 \leq T-1$), the duropolist would be able to extract all consumer surplus. Thus, pacman is an optimal strategy.
We now analyze the case where $M > T$.
Observe that because consumers will buy in the first period if and only if the price is not above their value, the duropolist's strategy space  can be restricted, without loss of generality, to announcing a first price equal to the valuation of some consumer. Let $\Pi(k)$ denote the profits for the whole game if the first price is $w_k$. Since by induction hypothesis the pacman strategy is an equilibrium after the first period, we have that
$$\Pi(k) = \sum_{i=1}^kn_i \cdot w_k + \sum_{j=k+1}^{T+k-1}n_j\cdot w_j.$$
Now we want to show that
$$\Pi(1) \geq \Pi(k)$$ for all $k=1,\hdots,M$.
Consider first the case where $k \leq T$. By Lemma \ref{lem:pacman1}, by setting $\beta=T$, we have that for all $k=1,\hdots,T$ and all $i=1,\hdots,k-1$:
\begin{equation*}
0 \leq  n_i \cdot w_i - n_i \cdot w_k - n_{T+i}\cdot w_{T+i}
\end{equation*}
Summing over $i$ we obtain
\begin{eqnarray*}
0 &\leq&
 \sum_{i=1}^{k-1} (n_i \cdot w_i - n_i \cdot w_k - n_{T+i}\cdot w_{T+i})   \\
&=&\sum_{i=1}^{k-1} n_i \cdot w_i - \sum_{i=1}^{k-1} n_i \cdot w_k - \sum_{j=T+1}^{T+k-1}n_{j}\cdot w_{j}   \\
&=&\sum_{i=1}^{k-1} n_i \cdot w_i  + \sum_{i=k}^{T} n_i \cdot w_i  - \sum_{i=1}^{k-1} n_i \cdot w_k - \sum_{i=k}^{T} n_i \cdot w_i - \sum_{j=T+1}^{T+k-1}n_{j}\cdot w_{j}   \\
&=&\sum_{i=1}^{T} n_i \cdot w_i - \sum_{i=1}^{k-1} n_i \cdot w_k - \sum_{j=k}^{T+k-1}n_{j}\cdot w_{j}  \\
&=&\sum_{i=1}^{T} n_i \cdot w_i - \sum_{i=1}^{k} n_i \cdot w_k - \sum_{j=k+1}^{T+k-1}n_{j}\cdot w_{j}  \\
&=&\Pi(1) - \Pi(k)
\end{eqnarray*}
Thus $\Pi(1) \geq \Pi(k)$. Second, consider the case where $k > T$. By Lemma \ref{lem:pacman1}, setting $\beta=k >T $, we have that for all $i=1,\hdots,k-1$:
\begin{eqnarray*}
0
&\leq&  \sum_{i=1}^{T} \left(n_i \cdot w_i - n_i \cdot w_k - n_{k+i}\cdot w_{k+i}\right)  \\
&=& \sum_{i=1}^{T} n_i \cdot w_i - \sum_{i=1}^T n_i \cdot w_k - \sum_{j=k+1}^{T+k}n_{j}\cdot w_{j} \\
&=&\Pi(1) - \Pi(k)
\end{eqnarray*}
Thus, again, $\Pi(1) \geq \Pi(k)$. So there exists an equilibrium where the duropolist uses the pacman strategy.
\end{proof}

We are now ready to prove the following theorem which provides sufficient and necessary conditions for the existence of an equilibrium that extracts all consumer surplus.

\begin{theorem}[Pacman Theorem] \label{pacman_theorem}
Consider a duropoly game $\mathcal{G}$ with $M \leq N$ distinct valuations. There exists an equilibrium at which the duropolist extracts all the consumer surplus if and only if $M \leq T$ and $v_i = p_i$ for all $i \in [N]$.
\end{theorem}
\proof
Take a game $\mathcal{G}$ with $M \leq T$ and $v_i = p_i$ for all $i \in [N]$.  By Lemma \ref{pacman_equilibrium_lemma_old_version} there exists an equilibrium in which the duropolist uses the pacman strategy.
Moreover, since $M \leq T$, under this equilibrium the duropolist obtains all the consumer surplus.

The contrapositive also holds. First, assume $M > T$. Then, since the number of time periods is less than the number of different valuations it is impossible for the
duropolist to extract the value of every consumer before the end of the game. Second, assume $v_i > p_i$ for some $i \in [N]$.
Now take an equilibrium that extracts all the consumer surplus. At equilibrium, prices must be non-increasing over time.
Moreover, since all consumer surplus is extracted, it implies that consumers also purchase in decreasing order of value over time which means that the skimming property holds. If consumer $i$ bought in the last period it means she has the lowest valuation, i.e. $v_i=w_M$, and $p_i=v_i$ which is a contradiction. In the case consumer $i$ buys before the last period, Lemma~\ref{lem_bound_price} (see next section) implies that consumer $i$ never pays more than $p_i$, again a contradiction.
\qed

We conclude this section with some observations. Recall, from Footnote~$2$, that there exist finite time horizon games in which the Pacman solution
has profits that are a factor $\log N$ greater than static monopoly profits. On the other hand, the
main result of next section is that duropoly profits are approximately the static monopoly profits and never more than twice the static monopoly profits. Thus, when Pacman is an equilibrium, static monopoly profits are at least half of Pacman profits. To see this, observe that the condition $p_i=v_i$ for all $ i \in [N]$, implies that consumer valuations decrease exponentially. In these scenarios,  $n_1 \cdot v_1$ (which is static monopoly revenue) is at least half of the sum of all consumer valuations.

\section{A Relationship between Duropoly Profits and Static Monopoly Profits}\label{sec:results}

In this section, we will prove our main result: the profits of the duropolist in a skimming property-satisfying SPNE of the duropoly game are at least the profits of the corresponding static monopolist, but at most double.
The first result is straightforward.

\begin{lemma}
$\Pi^M \le \Pi^D$
\end{lemma}
\begin{proof}
The proof is by induction on the number of periods. The base case is trivial. Consider the (possibly sub-optimal) sales schedule where the duropolist sells at $p_1$ in period 1, and then follows an equilibrium strategy for the remaining subgame $\mathcal{G}'$. Let $k_1$ be the number of consumers sold to in period 1 under this schedule. Let $\Pi^M_\mathcal{G} = j \cdot v_{j} \equiv j \cdot p_1$ and $\Pi^M_{\mathcal{G}'}= (k-k_1)v_{k}$. Since $j =|\{i | v_i \geq v_j\}|$, $j \geq k_1$ (no one with value less than $v_j$ is willing to pay $v_j$). But $k = \arg\max_{i\geq k_1} (i-k_1)v_i$, therefore $(k-k_1)v_{k} \geq (j-k_1)v_{j}$. So
\begin{align*}
\Pi^D  \geq  k_1 p_1 + \Pi^D_{\mathcal{G}'}
 \geq  k_1 p_1 + \Pi^M_{\mathcal{G}'}
  \geq  k_1 p_1 + (j-k_1)v_j
 =  k_1 p_1 + (j-k_1)p_1
= \Pi^M,
\end{align*}
where, in the second inequality, we used the induction hypothesis.
%One strategy for the duropolist is to charge a very high price in periods $1$ to $T-1$
%and then charge the monopoly price $p_1$ in period $T$. As no-one would then buy before period $T$,
%this would generate profits of $\Pi^M$. This strategy is rejected only if the dynamic program
%for the profit maximizing duropolist finds a more profitable solution. \qed
\end{proof}

The second claim is more substantial.

\begin{theorem}\label{thm:main}
$\Pi^D \le (\Pi^M + p_1) \le 2\cdot \Pi^M$
\end{theorem}

In Section \ref{sec:tight} we show that Theorem \ref{thm:main} is in fact tight, i.e., duropoly profits can
be as close as desired to twice the static monopoly profits. Observe however, that in many situations Theorem \ref{thm:main} suggests that duropoly profits can be only slightly higher than static monopoly profits. For example, this happens when the value of the highest value consumer is negligible with respect to the static monopoly profits, i.e. $p_1 \le v_1 << \Pi^M$.

%Specifically, one may think as if consumer $i$ valuation in a game $\mathcal{G}$ with $N$ consumers is equal to $F(\frac{i}{N})$ where $F:[0,1] \to \maathbb{R}_{\geq 0}$ is non-increasing function which represents the consumer valuation function in a game with a continuum of consumers. Here $\Pi^M = \max_{i \in [N]} F(\frac{i}{N})\cdot i$ and therefore $\frac{|Pi^D}{\Pi^M} \leq \frac{v_1 + \Pi^M}{\Pi^M} = \frac{F(\frac{1}{N} +  \max_{i} F(\frac{i}{N} \cdot i)}{\max_{i} F(\frac{i}{N}\cdot i)}
%consider the following multi step process to resemble a game with a non-atomic consumers as close as possible. divide each consumer into two and

We prove Theorem \ref{thm:main} in two steps. To describe them, recall we have $N$ consumers with valuations
$v_1\ge v_2\ge \cdots \ge v_N$. Furthermore, $\mathcal{G}_i$ is the market consisting of
consumers $\{i, i+1, \dots, N\}$ and $p_i$ is the static monopoly price for the market $\mathcal{G}_i$. Each $p_i$ is equal to some consumer's value so we will define $y_i$ as the consumer with the smallest index such that $p_i=v(y_i)$.
First we show that
$\Pi^D \le \sum_{i=1}^{N} p_i$, and second we show that $\Pi^M +v_1 \ge \sum_{i=1}^{N} p_i$.
\begin{lemma}\label{thm:up-prices}
The maximum profit of the duropolist satisfies
$\Pi^D \le \sum_{i=1}^{N} p_i$.
\end{lemma}

%Thus $|A_1| = y_1\cdot v(y_1)$ is the static monopoly revenue for the whole game.

To prove this we require the following three lemmas.
% \begin{lemma}\label{subgame_price_sequence}
%$p_i \ge p_{i+1}$ for $i=1,\dots,N-1$
%\end{lemma}
%\begin{proof}
%Without loss of generality we can restrict to the case $i=1$. Let $a= y_1 $ and $ b=y_{2}$. Then $p_1 = v_a$ and $p_{2}=v_b$ and therefore $a \ge 1$ and $b \ge 2$. As a first case, consider that $a \le b$. Given that valuations are non-increasing, it follows that $p_1 = v_a \ge p_2 = v_b$.  As the second case, suppose that $a>b \ge 2$. We know that $$av_a \geq bv_b,$$ by definition of $a$. Now if $a>b $, in the game $\mathcal{G}_2$ the static monopolist has the option of selling to exactly the consumers in $[2,a]$. Therefore, as $v_b$ is the static monopoly price for $\mathcal{G}_2$, the game with consumers $[2,N]$, it follows that $$(b-1)v_b \ge (a-1)v_a.$$ Combining these two inequalities gives us $v_a \geq v_b$. But $a>b$ implies $v_b \geq v_a$, so we conclude that $v_a = v_b$. In either case, $v_a \geq v_b$. \end{proof}

\begin{lemma}\label{lem_bound_price}
In equilibrium, consumer $i$ never pays more than $p_i$ whenever she buys before the last period.
\end{lemma}
\begin{proof}
We proceed by induction in the number of time periods. For $T=2$, let $A$ be the set of consumers that buy at $t=1$. Suppose for the purpose of contradiction that some consumer $i \in A$ pays a price higher than $p_i$. Because of the skimming property, we have that $A=\{1,\dots,k\}$ for some $k$. By Lemma \ref{subgame_price_sequence}, this price is also  more than $p_k$, so consumer $k$ pays more than $p_k$. But if consumer $k$ refuses to buy at $t=1$, then at $t=2$ the duropolist would charge $p_k$ which is a contradiction since consumer $k$ would have obtained a higher profit by waiting.

Now suppose that the lemma is true for all games of 1 to $T$ periods and consider a game $\mathcal{G}$ of $T+1>2$ periods. Let $E$ denote a subgame perfect Nash equilibrium with the skimming property in $\mathcal{G}$. For the purpose of contradiction, suppose that at some time period $t$ ($t<T+1$), there is at least one consumer $j$ that pays more than $p_j$. Among all those consumers, let $i$ denote the one with the lowest valuation. Since it is not true that at period $t$ consumer $\ell$ pays more than $p_{\ell}$ for $\ell >i$, and by Lemma \ref{subgame_price_sequence} $p_i \geq p_{\ell}$, then consumer $i$ is the lowest valuation consumer that buys at period $t$. Therefore, if consumer $i$ refuses to buy at period $t$ we end in the market $\mathcal{G}_i$ with $T+1-t$ periods. If $T+1-t=1$ (i.e. $t$ was the second to last period), the duropolist will charge the price $p_i$ at the last period. If $T+1-t>1$, it holds by the induction hypothesis that consumer $i$ would never pay more than $p_i$. Thus, we can conclude that such equilibrium $E$ cannot exist as consumer $i$ would have obtained a higher profit by waiting.
\end{proof}

\begin{lemma}\label{lem:mM}
The maximum revenue of the duropolist satisfies
$$\Pi^D \le \max_{m\le N}\, \left((y_m-m+1)\cdot v(y_m) + \sum_{i=1}^{m-1} p_i\right)$$
\end{lemma}
\begin{proof}
Consider the subgame perfection conditions. In the final time period $T$, consumer $i$ is willing to pay up to $v_i$.
In periods $1$ to $T-1$, consumer $i$ is willing to pay up $p_i=v(y_i)$, the static monopoly price for the market $\mathcal{G}_i$.

Suppose that in the optimal solution, the duropolist sells to consumers $\{m, m+1,\dots, M\}$ where $1\le m \le M\le N$ in the final
period $T$. Since consumer $M=y_m$ buys in the final period, the revenue then is exactly $(y_m-m+1)\cdot v(y_m)$.
By Lemma \ref{lem_bound_price}, consumers who buy in earlier periods, that is consumers $\{1,2,\dots m-1\}$, pay at most their static monopoly prices. Therefore,
the maximum revenue is upper bounded by
$$(y_m-m+1)\cdot v(y_m)+ \sum_{i=1}^{m-1} p_i$$
The result follows by taking the maximum over all consumers $m$.
\end{proof}

\begin{lemma}\label{lem:End} The static monopoly revenue for the market $\mathcal{G}_m$ is at most $$\sum_{j=m}^N p_j$$
\end{lemma}
\begin{proof}
Without loss of generality, by re-indexing so that $m=1$, it suffices to show that
\begin{equation} \label{eq:A1}
\Pi^M = y_1\cdot v(y_1) \le \sum_{j=1}^N p_j
\end{equation}
Let $C= \{ p_j: j = 1,\dots,N \}$. We proceed by induction on $|C|$, that is, the number of distinct static monopoly prices
over all the markets $\mathcal{G}_j$. For the base case, $|C|=1$, we have $p_1=p_j$ for all $j$. Thus $p_1=p_N=v_N$ and $y_1=N$.
Every consumer then pays $v_N$ and the total revenue is
$$y_1\cdot v(y_1)= N\cdot v_N = \sum_{j=1}^N p_j$$

Assume the proposition holds for $|C|=k-1 \geq 1$. Now take the case $|C|=k$.
Let consumer $l$ be the highest index consumer in the original game with $p_{l} = p_1$.
Thus $p_{l+1} < p_1 = v(y_1)$.
By the induction hypothesis, applied to the market $\mathcal{G}_{l+1}$ on consumers $\{l+1,l+2,\dots,N\}$,
we have
\begin{eqnarray}\label{eq:induction_1}
\sum_{i=l+1}^N p_i &\geq& (y_{l+1}- l) \cdot v(y_{l+1})
\end{eqnarray}
Consequently,
\begin{eqnarray*}
\sum_{i=1}^N p_i &=&  l \cdot p_1+ \sum_{i=l+1}^N p_i \\
&\geq& l \cdot p_1  + (y_{l+1}- l) \cdot v(y_{l+1}) \\
&\geq&  l \cdot p_1  + (y_{l}- l)\cdot v(y_{l}) \\
&=& l \cdot v(y_l)+ (y_{l}- l)\cdot v(y_{l})\\
&=& y_{l} \cdot v(y_l) \\
&=& y_{1} \cdot v(y_1)
\end{eqnarray*}
Here the first equality follows by definition of $l$. The first inequality follows by
(\ref{eq:induction_1}). The second inequality holds as $v(y_{l+1})$ is
the static monopoly price for the market $\mathcal{G}_{l+1}$.
The final three equalities follow by definition of $l$. That is $p_1=p_l$ and so $y_l = y_1$.

This shows that (\ref{eq:A1}) holds as desired.
\end{proof}

\noindent{\em Proof of Lemma \ref{thm:up-prices}.}
Combine Lemma \ref{lem:mM} and Lemma \ref{lem:End}. \qed

%In the last period of sales, it is a dominant strategy for the monopolist to sell at the static monopoly
%price of the subgame with the players who have not yet bought, and it is a dominant strategy for the
%players to accept any price which is less than or equal to their value.
%
%
%Lemma \ref{lemma:End} tells
%us that this amount is bounded above by the sum of the $p_i$ for all remaining players $i$.
%
%In earlier periods, no player $i$ can be charged more than $p_i$, since he, and all players below him,
% can simply wait until the last period, where they know the monopolist will charge no more than $p_i$.
% Thus in earlier periods, it is natural to seek a bound related to the $p_i$'s.
%
% We first find an upper bound for the static monopoly revenue $|A_1|$

It remains to prove the upper bound on $\sum_{i=1}^N p_i$.

\begin{lemma}\label{thm:down-prices}
$\sum_{i=1}^N p_i \le \Pi^M +p_1$
\end{lemma}
\begin{proof}
We proceed by induction on $N$. For games with a single consumer, the statement is trivially true.
Recall that $p_i=v(y_i)$ is the static monopoly price for the market $\mathcal{G}_i$ on consumers $\{i, i+1, \dots, N\}$ and $y_i$ is the index of the lowest value consumer whose value is not less than $p_i$.
Consider now a game $\mathcal{G}$ with $N+1$ consumers. It remains to show that $$\sum_{i=1}^{N+1} p_i \le v(y_1) +  y_1 \cdot v(y_1).$$

We proceed as follows:
\begin{eqnarray}\label{eq:plug}
\sum_{i=1}^{N+1} p_i
& = & v(y_1) + \sum_{i=2}^{N+1} p_i \nonumber\\
& \leq & v(y_1) + v(y_2) + (y_2-1) \cdot v(y_2)  \label{thm:induction_hypothesis}\\
& = & v(y_1) +  y_2 \cdot v(y_2)  \nonumber\\
& \le & v(y_1) +  y_1 \cdot v(y_1) \label{thm:optimality_p_1}
\end{eqnarray}
Here equation (\ref{thm:induction_hypothesis}) follows by the induction hypothesis and inequality (\ref{thm:optimality_p_1}) comes from the fact that $v(y_1)$ is the static monopoly price of the market $\mathcal{G}_1$.
\end{proof}

\noindent {\em Proof of Theorem \ref{thm:main}.}
The result follows from Lemmas \ref{thm:up-prices} and \ref{thm:down-prices}. \qed

%\subsection{Discontinuity in the limit} \label{sec:discontinuity}
%
%The factor $2$ bound is tight. There are examples in which the duropolist can extract twice as
%much profit as the equivalent static monopolist. To see this, assume
%that $T=2$ and set $v_j=v_H$ for $1\le j \le k$ and $v_j=v_L= \frac{1}{n-k+1}\cdot v_H$ for all $k+1\le j\le n$.
%The optimal solution is to charge $V_H$ in the first period and $V_L$ in the
%last period. The high value consumers will buy in the first period and the
%low value consumers in the last period. It can be verified that this solution satisfies
%the equilibria conditions.
%
%The total revenue is then
%\begin{eqnarray*}
%k\cdot v_H + (n-k)\cdot v_L &=&
%k\cdot v_H +  \frac{n-k}{n-k+1} \cdot v_H \\
%&=& |A_1| + \left(1-\frac{1}{n-k+1}\right)\cdot v_H
%\end{eqnarray*}
%Thus in the limit $n \gg k$ we obtain $v_{\max}+ |A_1|$. In the case $k=1$ we have that $|A_1|= v_{H}$
%and, thus, duropoly profits are twice static monopoly profits.

\subsection{A Tight Example} \label{sec:tight}

The factor $2$ bound is tight. There are examples in which the duropolist can extract twice as
much profit as the equivalent static monopolist. To see this, assume
that $T=2$ and set $v_j=v_H$ for $1\le j \le k$ and $v_j=v_L= \frac{1}{n-k+1}\cdot v_H$ for all $k+1\le j\le n$.
The optimal solution is to charge $V_H$ in the first period and $V_L$ in the
last period. The high value consumers will buy in the first period and the
low value consumers in the last period. It can be verified that this solution satisfies
the equilibria conditions.

The total revenue is then
\begin{eqnarray*}
k\cdot v_H + (n-k)\cdot v_L &=&
k\cdot v_H +  \frac{n-k}{n-k+1} \cdot v_H \\
&=& \Pi^M + \left(1-\frac{1}{n-k+1}\right)\cdot v_H
\end{eqnarray*}
Thus in the limit $n \gg k$ we obtain $\Pi^M+v_{H}$. As $v_H=p_1$, the
duropoly profits exceed static monopoly profits by an additive amount approaching static monopoly profits.
Moreover, in the case $k=1$ we have that $\Pi^M= v_{H}$
and, thus, duropoly profits are twice static monopoly profits.

\section{Non-Skimming Equilibria}\label{sec:NoSkimSum}

In this section, we discuss equilibria that do not satisfy the skimming property. For non-atomic consumers, it is easy to
show that such  equilibria cannot exist. However, for atomic consumers, subgame perfect Nash equilibria that do not satisfy the
skimming property can exist. To see this, consider the following three-consumer, two period problem.\footnote{One can also show that a non-skimming equilibrium exists in this example if we introduce a discount factor with value strictly less than 1.}

\begin{table}[h]
\centering
\caption{An Example with a Non-Skimming Equilibrium}
\begin{tabular}{|M|M|M|}

\hline
Consumer valuation & Period 1 Strategy & Period 2 Strategy \\
\hline
80 & Buy if $\mu_1 \le 45$ & Buy if $\mu_2 \le 80$\\
\hline
70 & Buy if $\mu_1 \le 70$ & Buy if $\mu_2 \le 70$\\
\hline
45 & Buy if $\mu_1 \le45$ & Buy if $\mu_2 \le 45$\\
\hline
\hline
Monopolist & $\mu_1=70$ & Charge Static Monopoly Price (for remaining consumers)\\
\hline
\end{tabular}
\label{ExSkim}
\end{table}

Because period 2 is the last period of the game, all the consumers who did not buy in period~1 have a dominant strategy
which is ``buy iff $\mu_2 \le v$". The duropolist's best response to this
is to charge a value $\mu_2$ equal to the static monopoly price for the subgame with all consumers who
did not buy in period 1. It is simple to check that all of the consumers' strategies are mutual best responses. The period 1 price is also the price which maximizes revenue for the duropolist, given the consumer strategies.
So this is a subgame perfect Nash equilibrium. However, in this equilibrium, the consumer with value 70 buys in period 1,
while the consumer with value 80 buys in period 2 (where $p_2 = 45$). Therefore the skimming property is not satisfied. Also note that the consumer with the second highest value pays his full value, which is significantly more than $p_2$, the static monopoly price in the market $\mathcal{G}_2$ ($p_2=45$).

Note, however,  that if we swapped the period 1 strategies of the top two consumers, we would still have an
equilibrium, and one which satisfies the skimming property. We can prove that this is always the case.

\begin{theorem}\label{lemma:2Skim}
In a 2 period game, for every strategy equilibrium which does not satisfy the skimming property, there is a corresponding equilibrium at the same prices which satisfies the skimming property and earns the same revenue for the duropolist.
\end{theorem}

To prove this theorem, we first require some results about static monopoly prices.

%\begin{claim}\label{claim:smp1}
% Let $V= \{v_1,v_2,\hdots,v_n\}$ be a set of valuations $v_1 \geq v_2 \hdots \ge v_n$.
% Let $k = \arg \max_i iv_i$ and $p_V^{sm} =v_k$. Now let $V'$ be the set $V$ with $v_1$
% replaced by a value $v_2< v'< v_1$, and re-index the set so $v' \equiv v'_1 \geq v_2 \equiv v'_2 \hdots \geq v_n \equiv v'_n$.
% Then $p_{V}^{sm} \geq p_{V'}^{sm}$.
%\end{claim}
%
%\begin{proof}
%Note that $v'_i \leq v_i$ for every $i$. If $k >1$, then it is clear that $kv'_k \equiv kv_k \geq iv_i \geq iv'_i$ for
%all $i$. Therefore, either $p_{V'}^{sm} = p_{V}^{sm} = v_k$. If $k=1$, then since $v'_i < v_1$ for all $i$, we
%have that $p_{V}^{sm} \geq p_{V'}^{sm}$, no matter what value the arg max takes.
%\end{proof}

\begin{claim}\label{claim:smp1}
Let $V= \{v_1,v_2,\hdots,v_N\}$ be a set of valuations $v_1 \geq v_2 \hdots \ge v_N$.
Let $k = \arg \max_i iv_i$ and $p_V^{sm} =v_k$, the static monopoly price calculated for the set $V$. Now let $V'$ be the set $V$ with some $v_j$ replaced by a value $v_j > v' \geq v_k$, and re-index the set so $v'_1 \geq v'_2 \hdots \geq v'_n$.  Then $p_{V'}^{sm} = v_k$.
If, instead, $v' \geq v_j > v_k$ and we re-index the set so $v'_1 \geq v'_2 \hdots \geq v'_N$, then we still have $p_{V'}^{sm} \geq v_k$.
\end{claim}

\begin{proof}
First suppose $v' <v_j$. Note that $v'_i \leq v_i$ for every $i$. Therefore, by definition of $k$, $kv_k \geq iv_i \geq iv'_i$ for all $i$.
For all $i$ such that $v_i \leq v'$, we still have $v'_i = v_i$, so therefore $v_k = v'_k$ and $kv'_k \geq iv'_i$ for all $i$. Therefore $k = \arg \max_i iv'_i$ and $p_{V'}^{sm} = v_k$.

Next assume $v' \geq v_j$. Note that, because $v_j >v_k$, we still have $v_i '=  v_i$ for every $i \geq k$, and as $kv_k \geq iv_i$ for all $i \geq k$,
 it follows that the index which maximizes $lv_l$ can be no larger than $k$. Therefore $p_{V'}^{sm} \geq v'_k = v_k$.
\end{proof}

\begin{claim}\label{claim:smp2}
Let $V= \{v_1,v_2,\hdots,v_N\}$ be a set of valuations $v_1 \geq v_2 \hdots \ge v_N$.
Let $k = \arg \max_i iv_i$ and $p_V^{sm} =v_k$. Now add a value $v' \ge v_k$ to the set of valuations, call the new set $V'$
and re-index the set so  $v'_1 \geq v'_2 \hdots \geq v'_{N+1}$. Then $p_V^{sm} \le p_{V'}^{sm} \leq v'$.
\end{claim}

\begin{proof}
Let $v' = v'_l$ (i.e. $v'$ is the $l$-th highest value in $V'$). We have that $v'_i = v_{i-1}$ for $i>l$. The result follows immediately as for all $j <l$, $j v_j \leq k v_k < (k+1) v_k = (k+1) v'_{k+1}$, so regardless of what the maximizing index actually is, it cannot be less than $l$, and therefore $p_{V'}^{sm} \leq v_l =v'$. But we also have that $(k+1) v'_{k+1} = (k+1)v_k \geq (j+1)v_j = (j+1) v'_{j+1}$ for all $j \geq k$. Therefore the maximizing index cannot be more that $k+1$, and thus $v'_{k+1} = v_k = p_V^{sm} \leq p_{V'}^{sm}$.
\end{proof}

Consider an equilibrium $\mathcal{E}$ which violates the skimming property. The violation must occur in period 1, as in the second period the dominant strategy followed by the consumers ensures that there is a cut-off value above which all remaining consumers buy. We must then have at least one pair of consumers which we denote by their values $v,w$ with $w>v$, where $v$ buys in
period 1  and $w$ does not buy in period 1. Without loss of generality, let  $w$ be the highest valued consumer who
doesn't buy in period 1, and $v$ be the lowest
valued consumer who buys with value below $w$. We denote by $E$ the set of consumers under the equilibrium $\mathcal{E}$ who did
not buy in period 1. Then $w \in E$ and $v \notin E$.

The period 2 price, $\mu_2(E)$, is fixed by the fact that
it is a dominated strategy to charge the static monopoly price of the remaining consumers.
For a consumer $y\notin E $ who buys in period 1, we will also need to compare $\mu_1(E)$ against its  threat price, $\mu_2(E^y)$.
Here $E^y = E \cup \{y\}$, and so $\mu_2(E^y)$ is the period 2 price $y$ would face if he did not buy.

\begin{lemma}\label{lem:eq}
At the equilibrium $\mathcal{E}$, we have
\begin{equation}
w >v \ge \mu_2(E^v) \ge \mu_1(E)  \ge \mu_2(E)
\end{equation}
\end{lemma}
\begin{proof}
By assumption $w>v$.
Consumer $v$ wants to buy in the first period.
The equilibrium conditions then imply that (i) $\mu_1(E)\le  v$ and (ii) $\mu_1(E)\le \mu_2(E^v)$.
Consumer $w$ does not want to buy in the first period. But since the first period price is less than his value,
he would be willling to buy in the first period if $\mu_2(E) > \mu_1(E)$. The equilibrium conditions then imply
that $\mu_2(E)\le \mu_1(E) <w$ (and $w$ does buy in the second period).
It only remains to prove that $v \ge \mu_2(E^v)$. We have seen that  $v \ge \mu_2(E)$.
But then, by Claim \ref{claim:smp2}, we obtain $v \ge \mu_2(E^v)$.
\end{proof}

To prove Theorem \ref{lemma:2Skim}, we will show that with the same period 1 price, the strategy
profile $\mathcal{S}$ which consists of $v$ and $w$
swapping actions from $\mathcal{E}$ (and all other consumer strategies remaining the same) is also an equilibrium.
We will also show that the period 2 price stays the same after swapping. By repeated
swapping, we will remove all pairs where the lower valued consumer buys before the higher valued
consumer, until  we reach an equilibrium where the skimming property is satisfied, while the prices are unchanged.
So let $S = \{E - w\} \cup \{v\}$ be the set of consumers who did not buy in period 1 under the swapped strategy $\mathcal{S}$.

\begin{lemma}\label{lem:eq2}
For the strategy profiles $\mathcal{E}$ and $\mathcal{S}$, we have
\begin{equation}
w >v \ge \mu_2(E^v) = \mu_2(S^w) \ge \mu_1(E) = \mu_1(S) \ge \mu_2(E) = \mu_2(S)
\end{equation}
\end{lemma}
\begin{proof}
The inequalities follow from Lemma \ref{lem:eq}. Thus it remains to prove the three equalities.
Since $S$ and $E$ differ in exactly the pair $\{w, v\}$ we have that $S^w=E^v$.
The equilibrium constraints imply the duropolist applies static monopolist pricing in the
final period; thus $\mu_2(S^w)=\mu_2(E^v)$.
Next, in the proposed solution corresponding to $S$, the duropolist is, by choice, setting the
same price in period $1$ for $S$ as for $E$; thus $\mu_1(S)=\mu_1(E)$.
Finally, we know that $v\ge \mu_2(E)$. But, then, if we reduce the value of consumer $w$ to $v$,
the static monopoly price will not change by the first case of Claim \ref{claim:smp1}; equivalently
 $\mu_2(S)=\mu_2(E)$.
\end{proof}

\begin{lemma}\label{lem:swap-eq}
The solution corresponding to $S$ is an equilibrium.
\end{lemma}
\begin{proof}
We prove that the strategy profile consisting of $v$ and $w$ swapping actions and all other
consumers acting the same way is still an equilibrium by considering the best responses of
every consumer separately.

\noindent $\bullet$ {\bf Consumer $\mathbf{v}$}: Consumer $v$ buys in the second period for $S$. If this is a best response
strategy then we need $\mu_2(S)\le \min [v, \mu_1(S)]$. This is true, by Lemma \ref{lem:eq2}.

\noindent $\bullet$ {\bf Consumer $\mathbf{w}$}: Consumer $w$ buys in the first period for $S$. If this is a best response
strategy then we need $\mu_1(S)\le \min [w, \mu_2(S^w)]$. This is true, by Lemma \ref{lem:eq2}.

\noindent $\bullet$ {\bf A consumer $\mathbf{u\neq w}$ who buys in period 1}: We wish to show
that $u$ still wishes to buy in the first period after we swap: that is, $\mu_1(S) \le \min [u, \mu_2(S^u)]$.

Since $v$ is the lowest valued consumer who buys in period 1, we must have $u \geq v$.
We know, by Lemma \ref{lem:eq2} that $\mu_1(S)\le v\le u$. Thus it suffices to prove that
$\mu_1(S) \le \mu_2(S^u)$.
We now have two possibilities:
\begin{enumerate}
\item[(a)] {\bf u $<$ w}: By Lemma \ref{lem:eq}, we have
$\mu_2(E^v) \leq v$. Now set $V = E^v = S^w$ and apply the first case of Claim \ref{claim:smp1} with $v'=u$ and $v_j=w$.
Thus, $V'= S^u$ and $\mu_2(E^v) = \mu_2(S^u)$. By the equilibrium conditions, we know
$\mu_1(E)\le \mu_2(E^v)$. Therefore $\mu_1(S) = \mu_1(E) \le \mu_2(E^v)=\mu_2(S^u)$.

\item[(b)]{\bf  u $\geq$ w}:
By Lemma \ref{lem:eq2}, we have $\mu_2(E) = \mu_2(S)<w$.
Now set $V = S$ and apply Claim~\ref{claim:smp2} by  adding a new consumer with value
$v'=w$; thus  $V'= S^w$ and $\mu_2(S^w)\le w$.  Applying the second case of Claim \ref{claim:smp1}
with $V=S^u$ and $V' = S^w$ and $v'=u$, then gives $\mu_2(S^u)\ge \mu_2(S^w)$. Furthermore, by Lemma~\ref{lem:eq2},
we know $\mu_1(S)\le \mu_2(S^w)$; thus $\mu_1(S) \le \mu_2(S^u)$, as desired.
\end{enumerate}

\noindent $\bullet$ {\bf Any other consumer $\mathbf{u\neq \{w, v\}}$}: Such a consumer $u$ either buys in period 2
or does not buy at all. In the former case, $\mu_2(E)\le \min [u, \mu_1(E)]$.
In the latter case, $u \le \mu_2(E) \le \mu_1(E)$. But, by Lemma \ref{lem:eq2}, we have
$\mu_1(S)=\mu_1(E)$ and $\mu_2(S)=\mu_2(E)$. Therefore, in either case, the same best response
conditions hold for $S$ as they did for $E$.
\end{proof}

\begin{proof}[Proof of Theorem \ref{lemma:2Skim}:]
By Lemma \ref{lem:swap-eq}, we know that $\mathcal{S}$ gives an equilibrium.
From Lemma \ref{lem:eq2}, we have that $\mu_1(S) = \mu_1(E)$ and $\mu_2(S) = \mu_2(E)$.
Since the same number of consumers buy in each period under $\mathcal{S}$ and $\mathcal{E}$, the duropolist earns the same revenue in $\mathcal{S}$ as in $\mathcal{E}$. By repeatedly
swapping pairs which are out of order, we obtain an equilibrium with the same revenue which
satisfies the skimming property.
\end{proof}

We conjecture that our two-periods result extends to the $T$-periods case, and therefore our bounds would apply not just to skimming-property satisfying equilibria, but to all sub-game perfect Nash equilibria.

%\begin{conjecture}
%The profits of the duropolist in any SPNE are at least the profits of the corresponding static monopolist, but at most double.
%\end{conjecture}

\section{Conclusions}
In this paper we studied the durable good monopoly problem, a classical problem in bargaining theory. In our setting, we consider consumers to be atomic and that there is a finite time horizon during which sales occur. We characterized all profit maximising strong-Markovian equilibria and proved that, in those equilibria, duropolist profits are comparable to those of a static monopolist. This is in contrast with previous results in which duropoly profits are either arbitrary small or arbitrary large compared to those of a static monopolist.

The paper leaves interesting questions for future research. We provided the first example of a sub-game perfect equilibrium in which the classical skimming-property does not hold. Although such equilibria might be rare in practice, we proved that for two-period games, the duropolist profits are not larger in those non-skimming equilibria than in the skimming equilibria. Whether the duropolist profits are approximately equal to the static monopoly profits (and never more than twice that) under non-skimming equilibria for games with an arbitrary number of time periods is a question we believe very worthy of study.

Finally, in this paper we studied the case in which there is no discount factor, neither for the duropolist, or for the consumers. However, one could easily extend the equilibrium characterization we provided in Section 3 in order to capture the monopolist and consumers behaviour when there is a discount factor. Moreover, under a setting with discounts, it is still true that consumers have more bargaining power than under a infinite time horizon (it is a simple exercise to construct an example with discounts where the Pacman strategy is not an equilibrium). But the specific question about how much profits can a duropolist obtain with respect to the static monopoly profits under a subgame perfect Nash equilibrium remains open.

%Specifically we show that duropoly profits are at most the static monopoly profits plus the static monopoly price regardless of the number of consumers, their values, and the number of time periods.

\ \\ \noindent{\bf Acknowledgements.}\ \\
%This work was partially supported by an NSERC posdoctoral fellowship award. This support is gratefully acknowledged.
We are very grateful to Juan Ortner, Jun Xiao, Anton Ovchinnikov and Sven Feldmann for their helpful comments and suggestions. %(to add later) The authors are also very grateful to an anonymous conference referee for suggestions that greatly simplified the proof of Lemma 5.4. \newpage

\bibliography{references_duropoly}

\begin{thebibliography}{17}
\providecommand{\natexlab}[1]{#1}
\providecommand{\url}[1]{\texttt{#1}}
\expandafter\ifx\csname urlstyle\endcsname\relax
  \providecommand{\doi}[1]{doi: #1}\else
  \providecommand{\doi}{doi: \begingroup \urlstyle{rm}\Url}\fi

\bibitem[Ausubel and Deneckere(1989)]{ausubel1989reputation}
Lawrence~M Ausubel and Raymond~J Deneckere.
\newblock Reputation in bargaining and durable goods monopoly.
\newblock \emph{Econometrica}, 57\penalty0 (3):\penalty0 511--531, 1989.

\bibitem[Bagnoli et~al.(1989)Bagnoli, Salant, and
  Swierzbinski]{bagnoli1989durable}
Mark Bagnoli, Stephen~W Salant, and Joseph~E Swierzbinski.
\newblock Durable-goods monopoly with discrete demand.
\newblock \emph{The Journal of Political Economy}, 97\penalty0 (6):\penalty0
  1459--1478, 1989.

\bibitem[Cason and Sharma(2001)]{cason2001durable}
Timothy~N Cason and Tridib Sharma.
\newblock Durable goods, {C}oasian dynamics, and uncertainty: Theory and
  experiments.
\newblock \emph{Journal of Political Economy}, 109\penalty0 (6):\penalty0
  1311--1354, 2001.

\bibitem[Coase(1972)]{coase1972durability}
Ronald~H Coase.
\newblock Durability and monopoly.
\newblock \emph{Journal of Law and Economics}, 15:\penalty0 143, 1972.

\bibitem[Cramton and Tracy(1992)]{cramton2002strikes}
Peter Cramton and Joseph~S Tracy.
\newblock Strikes and holdouts in wage bargaining.
\newblock \emph{The American Economic Review}, 82\penalty0 (1):\penalty0
  100--121, 1992.

\bibitem[Fuchs and Skrzypacz(2011)]{fuchs2011bargaining}
William Fuchs and Andrzej Skrzypacz.
\newblock Bargaining with deadlines and private information.
\newblock Technical report, Mimeo, University of California Berkley, 2011.

\bibitem[Fudenberg and Tirole(1991)]{tirole}
Drew Fudenberg and Jean Tirole.
\newblock \emph{Game Theory}.
\newblock MIT Press, 1991.

\bibitem[Fudenberg et~al.(1985)Fudenberg, Levine, and Tirole]{fudenberg1985}
Drew Fudenberg, David Levine, and Jean Tirole.
\newblock Infinite horizon models of bargaining with one-sided incomplete
  information.
\newblock In A.~Roth, editor, \emph{Game Theoretic Models of Bargaining}.
  Cambridge University Press, 1985.

\bibitem[Gul et~al.(1986)Gul, Sonnenschein, and Wilson]{gul1986foundations}
Faruk Gul, Hugo Sonnenschein, and Robert Wilson.
\newblock Foundations of dynamic monopoly and the coase conjecture.
\newblock \emph{Journal of Economic Theory}, 39\penalty0 (1):\penalty0
  155--190, 1986.

\bibitem[Guth and Ritzberger(1998)]{guth1998durable}
Werner Guth and Klaus Ritzberger.
\newblock On durable goods monopolies and the {C}oase-conjecture.
\newblock \emph{Review of Economic Design}, 3\penalty0 (3):\penalty0 215--236,
  1998.

\bibitem[McAfee and Wiseman(2008)]{mcafee2008capacity}
R~Preston McAfee and Thomas Wiseman.
\newblock Capacity choice counters the coase conjecture.
\newblock \emph{The Review of Economic Studies}, 75\penalty0 (1):\penalty0
  317--331, 2008.

\bibitem[Montez(2013)]{montez2013inefficient}
Jo{\~a}o Montez.
\newblock Inefficient sales delays by a durable-good monopoly facing a finite
  number of buyers.
\newblock \emph{The RAND Journal of Economics}, 44\penalty0 (3):\penalty0
  425--437, 2013.

\bibitem[Orbach(2004)]{orbach2004durapolist}
Barak Orbach.
\newblock The duropolist puzzle: Monopoly power in durable-goods market.
\newblock \emph{Yale Journal on Regulation}, 21:\penalty0 67--118, 2004.

\bibitem[Ortner(2016)]{ortner2016durable}
Juan Ortner.
\newblock Durable goods monopoly with stochastic costs.
\newblock \emph{Theoretical Economics}, \penalty0 (Forthcoming), 2016.

\bibitem[Stokey(1979)]{stokey1979intertemporal}
Nancy~L Stokey.
\newblock Intertemporal price discrimination.
\newblock \emph{The Quarterly Journal of Economics}, 93\penalty0 (3):\penalty0
  355--371, 1979.

\bibitem[von~der Fehr and Kuhn(1995)]{von1995coase}
Nils-Henrik~Morch von~der Fehr and Kai-Uwe Kuhn.
\newblock Coase versus pacman: Who eats whom in the durable-goods monopoly?
\newblock \emph{Journal of Political Economy}, 103\penalty0 (4):\penalty0
  785--812, 1995.

\bibitem[Williams(1983)]{williams1983legal}
Gerald~R Williams.
\newblock \emph{Legal negotiation and settlement}.
\newblock West Publishing Company, 1983.

\end{thebibliography}

\section*{Appendix 1: Proofs omitted from main text.}
     \setcounter{lemma}{0}\addtocounter{section}{-6}

    \renewcommand{\thelemma}{\Alph{section}\arabic{lemma}}
    \setcounter{conjecture}{0}
    \renewcommand{\theconjecture}{\Alph{section}\arabic{conjecture}}
     \setcounter{theorem}{0}
    \renewcommand{\thetheorem}{\Alph{section}\arabic{theorem}}

\begin{lemma}[Lemma \ref{subgame_price_sequence}]
Let the consumers in $\mathcal{G}$ be ordered by value so $v_i \geq v_{i+1}$ for all $i$. Then static monopoly prices are non-increasing in $i$: $p_i \ge p_{i+1}$ for $i=1,\dots, N-1$.
\end{lemma}

\proof
Without loss of generality we can restrict to the case $i=1$. Let  $p_1 = v_a$ and $p_{2}=v_b$ and therefore $a \ge 1$ and $b \ge 2$. As a first case, consider that $a \le b$. Given that valuations are non-increasing, it follows that $p_1 = v_a \ge p_2 = v_b$.  As the second case, suppose that $a>b \ge 2$. We know that $$av_a \geq bv_b,$$ by definition of $p_1$. Now if $a>b $, in the market $\mathcal{G}_2$ the static monopolist has the option of selling to exactly the consumers in $[2,a]$. Therefore, as $v_b$ is the static monopoly price for $\mathcal{G}_2$, the game with consumers $[2,N]$, it follows that $$(b-1)v_b \ge (a-1)v_a.$$ Combining these two inequalities gives us $v_a \geq v_b$. But $a>b$ implies $v_b \geq v_a$, so we conclude that $v_a = v_b$. In either case, $v_a \geq v_b$.\newline \qed

\begin{lemma} [Lemma \ref{lem:threatmono}]
In any game $\mathcal{G}$ with $T$ periods, for all $i \leq k$ and all $t<T$, $\tau(i,t) \geq \tau(k,t)$.
\end{lemma}

\proof We proceed by backwards induction on $t$. For $t = T-1$, $\tau(i,t)=p_i$, the static monoply price for the market $\mathcal{G}_i$, for all $i$. By Lemma \ref{subgame_price_sequence}, $p_i \geq p_k$ whenever $i \leq k$. Now consider any earlier period $t$. $p(i,t+1) = p(j^*(i,t+1),t+2)$ for all $i$. By the inductive hypothesis, if $j^*(i,t+1) \leq j^*(k,t+1)$, then we are done. Recall that $j^*(k,t+1)$ is determined by the sales schedule which maximizes revenue earned in $\mathcal{G}_i$ for the remaining periods. In particular, for all $l \geq j^*(k,t+1)$,

\begin{align*} (j^*(k,t+1)-k+1)\cdot p(j^*(k, t+1) , t+2) + \Pi(j^*(k, t+1)+1, t+2)   \\   \geq  (l-k+1)\cdot p(l , t+2) + \Pi(l+1, t+2)
\end{align*}

Again, by the inductive hypothesis, $p(j^*(k, t+1) , t+2) \geq p(l , t+2)$. Now multiply this inequality by $k-i$ (which is non-negative) and add it to the above to get

\begin{align*}(j^*(k,t+1)-i+1)\cdot p(j^*(k, t+1) , t+2) + \Pi(j^*(k, t+1)+1, t+2)   \\   \geq  (l-i+1)\cdot p(l , t+2) + \Pi(l+1, t+2),
\end{align*}

for every $l \geq j^*(k,t+1)$. Since $j^*(i,t+1)$ satisfies

$$j^*(i, t+1) = \arg \max_{j\ge i} \,\left((j-i+1)\cdot p(j, t+2)+ \Pi(j+1, t+2)\right)$$

it follows that $j^*(i,t+1) \leq j^*(k,t+1)$, which gives us our result. \qed

\begin{lemma}[Lemma \ref{lem:monotone}]
Consider two duropoly games, $\mathcal{G}$, with $T$ periods and a set $S$ of consumers, and $\mathcal{G'}$, with $T$ periods and a set $S'$ of consumers such that only the top valued consumers in $S$ and $S'$ differ and the top valued consumer in $S'$, who we name $x$, has the higher value. If we use $p^*_\mathcal{G} (1,1)$ and  $p^*_\mathcal{G'} (1,1)$ to denote the first period prices as calculated by the recursion relationship in Section \ref{sec:opt-cond}, then $p^*_\mathcal{G'} (1,1) \geq p^*_\mathcal{G} (1,1)$.
\end{lemma}

\proof It is easily seen that in the case $T=1$, the result holds (the price either remains the same or increases to $v_x$). Consider the optimal sales schedule for the game $\mathcal{G}$, and assume this schedule sells to more than one person in period 1. But prices for a schedule which sells to more than one person in period 1 do not depend on the value of the top consumer in $S$ (the price in period 1 depends on the threat price of a lower-valued consumer, and later prices depend only on the consumers left). Therefore, we can achieve the same profit from such a schedule in $\mathcal{G'}$ with the same prices. So in $\mathcal{G'}$, the optimal sales and pricing schedule either is the same as the optimal for $\mathcal{G}$, in which case we are done, or involves selling only to $x$ in the first period. If the durpolist sells to to $x$ in the first period, it is at a price  $p^*_\mathcal{G'} (1,1)$ equal to $x$'s threat price. This is, by definition, the same price as for the game $\mathcal{G''}$ with consumers $S'$ but with $T-1$ periods instead of $T$ periods. By the induction hypothesis, this price is higher than the corresponding optimal first period price $p_S$ for the game with consumers $S$ but with $T-1$ periods. But $p_S$, by definition, is the threat price for the top consumer in $S$ in $\mathcal{G}$, and therefore at least as high as the period 1 threat price under the optimal sales schedule in $\mathcal{G}$ (if $v_i  \geq v_j$, $i$'s threat price is $\geq j$'s threat price: see proof of Lemma \ref{lem:threatmono} in this section). But the optimal period 1 threat price is  $p^*_\mathcal{G} (1,1)$, so $p^*_\mathcal{G'} (1,1) \geq p_S \geq p^*_\mathcal{G} (1,1)$.

It remains to prove the case where the optimal sales schedule in $\mathcal{G}$ sells to just the top consumer in $S$. In this case, the price charged in period 1 is $p_S$. But by the same argument as above, the threat price for $x$ is at least as high as $p_S$. Therefore if the optimal sales schedule in $\mathcal{G'}$ sells to just $x$ in the first period, the first period price is at least as high as in $\mathcal{G}$. But note that if the duropolist sells to more than one person in period 1 of $\mathcal{G'}$, she achieve the same profit as a sub-optimal sales schedule in $\mathcal{G}$. But she can clearly beat that reveue by selling to $x$ in period 1 and then following the optimal sales schedule in $\mathcal{G}$ from period 2 onwards.  Therefore, whatever the optimal sales schedule in $\mathcal{G'}$, it must involve selling exactly one item in period 1. Therefore we have $p^*_\mathcal{G'} (1,1) \geq p_S = p^*_\mathcal{G} (1,1)$.

We have covered all cases, so the lemma is proved. \qed

\begin{lemma}[Lemma \ref{lem:EqPriceDec}]
In any game $\mathcal{G}$ with $T$ periods, if the duropolist and consumers follow the strategies described in Section \ref{sec:opt-cond}, then prices are non-increasing in time.
\end{lemma}

\proof If both duropolist and consumer follow the strategies described in Section \ref{sec:opt-cond}, the duropolist will select an initial sales path $\{x_t\}$ , such that, for the last consumer $j_t$ scheduled to buy in period t ($j_t = \sum_{i\leq t} x_i$), we have the recursive relationship:
$$p(j_{t} + 1, t+1) = \tau(j_{t+1},t+1) = p(j_{t+1},t+2)$$
In other words, in period $t$ the duropolist plans to sell to consumers $j_t+1, j_t+2, \dots, j_{t+1}$ at $j_{t+1}$'s threat price. By Lemma \ref{lem:threatmono}, this is less than or equal the threat price of everyone in the set $\{j_t+1, j_t+2, \dots, j_{t+1}\}$. So under the consumer strategies specified, all $x_t$ consumers in $\{j_t+1, j_t+2, \dots, j_{t+1}\}$ buy in period $t$, and the duropolist's strategy never deviates from the initial sales path. So the price she charges in each period $t$ is  $p(j_{t-1} + 1, t)$ (with $j_0 \equiv 0$).

As part of Lemma \ref{lem:threatmono} we showed that $j^*(i,t+1) \leq j^*(k,t+1)$ for all $i \leq k$. Using this and the result of Lemma 2, we have

\begin{align*}
p(j_{t-1} + 1, t) &&= & & p(j_{t},t+1) \\ && = && \tau(j^*(j_{t},t+1),t+1)  \\
 &&\geq &&\tau(j^*(j_t+1,t+1),t+1) \\
&& = &&\tau(j_{t+1},t+1) \\
&&= && p(j_t+1, t+1),
\end{align*}
 where in the fourth line, we use the fact that $j^*(j_t+1,t+1)= j_{t+1}$ from the definition of the $j_t$'s and the argmax condition of the recursion relation (\ref{eq:characterization_prices}). So these prices are non-increasing in time. \qed

\begin{lemma}\label{NeverWait}
There is a subgame perfect equilibrium which follows the recursion relationship (\ref{eq:characterization_prices}) in which a sale occurs in each period until all consumers have already bought the item, and this equilibrium achieves at least as much profit for the duropolist as any which allows the duropolist to not sell in some periods where there are consumers remaining.
\end{lemma}
\proof The proof will be by induction. In the case $T=1$, it is clearly a dominant strategy for the durpolist to sell if there are any consumers remaining, as otherwise he will earn nothing. This also holds for the last period of a longer game in a subgame perfect equilibrium.

Now consider $T>1$, and to start, assume that the claim is false. Then there must be a game $\mathcal{G}$ with a period $t^l<T$ such that the following two conditions hold: (a) there remain consumers who haven't bought at the start of period $t^l$, but the duropolist does sell any items in this period; (b) the duropolist sells items in every subsequent period of the game. So there is an equilibrium for a game $\mathcal{G'}$, corresponding to the subgame of $\mathcal{G}$ starting at $t^l$, where the duropolist sells nothing in the first period, and sells at least one item in all subsequent periods, until all consumers have bought, and this equilibrium yields strictly more profit than one which follows (\ref{eq:characterization_prices}). Thus, to show our claim is true, we only need to show that equilibria where we sell nothing in the first period followed by sales in every subsequent period (until all consumers have bought) do not yield more profit than those which follow the recursion relationship (\ref{eq:characterization_prices}). Furthermore, it is enough to show a sub-optimal sales schedule yields as much profit, as the optimal recursion relationship result must do even better.

Let $\mathcal{G}(N,T)$ be our game with a set of consumers $N$ and $T$ periods. If nothing is sold in the first period, and the duropolist sells at least one item in each subsequent period (until all consumers have bought), the duropolist can achieve profit of  at most $\Pi^D_{\mathcal{G}(N,T-1)}$ (see discussion preceeding Corollary \ref{cor:uniquemax}). Let $k \geq 1$ and $p$ be the number of items sold in the first period and the first period price, respectively, of $\mathcal{G}(N,T-1)$ under (\ref{eq:characterization_prices}). Consider the following (possibly sub-optimal) strategies for $\mathcal{G}(N,T)$: the duropolist sells at price $p$ in period 1 and follows the equilibrium of (\ref{eq:characterization_prices}) for all subsequent periods, while the consumers buy iff the price is less than their threat price. We know that the top consumer will buy in period 1 as she would be offered the same price if everyone refused to buy in period 1. We also know that no consumer $i>k$ will buy as $p$ is equal to $k$'s threat price, $\tau_{\mathcal{G}(N,T-1)}(k,1)$, but  if $i$ and all below her refused to buy in $\mathcal{G}(N,T)$, the price in the second period of $\mathcal{G}(N,T)$ would be at most $\tau_{\mathcal{G}(N,T-1)}(k+1,1)$. So some number $1 \leq l \leq k$ buys in the first period. By definition of (\ref{eq:characterization_prices}), we sell at least one item in each subsequent period as well, as long as there are consumers remaining to buy. If $l=k$, then, by the induction hypothesis, we make at least $\Pi^D_{\mathcal{G}(N-[k],T-2)}$  with sales in each subsequent period until all consumers are sold to. So our total profit from this possibly sub-optimal scheme is at least

$$\Pi^D_{\mathcal{G}(N-[k],T-2)} + k \cdot  p = \Pi^D_{\mathcal{G}(N,T-1)}.$$

If $1 \leq l<k$, we know that $\Pi^D_{\mathcal{G}(N-[l],T-1)}$ is larger than the profit obtained by selling to $k-l$ consumers at price $p$ and then following the equilibrium for $\mathcal{G}(N-[k],T-2)$ given by (\ref{eq:characterization_prices}). So we obtain profit of

$$ l \cdot p + \Pi^D_{\mathcal{G}(N-[l],T-1)} \geq l \cdot p + (k-l) \cdot p + \Pi^D_{\mathcal{G}(N-[k],T-2)}  = \Pi^D_{\mathcal{G}(N,T-1)}.$$

Therefore under the optimal schedule, the profit is at least $\Pi^D_{\mathcal{G}(N,T-1)}$. Therefore there is no benefit to waiting a period before starting to sell. This proves the claim. \qed

\begin{lemma} [Lemma \ref{lem:pacman1}]
If $p_i=v_i$ for all $i \in [N]$, the following inequality holds for every natural number $\beta \geq 2$, $k=2,\hdots,\beta$ and all $i=1,\hdots,k-1$,
$$n_i \cdot w_i - n_i \cdot w_k - n_{\beta+i}w_{\beta+i} \geq 0.$$
\end{lemma}
\begin{proof}
The statement is trivially true if $k>M$ as then $w_k=w_{\beta+i}=0$.  Thus, we may assume that
$1\le i <k \le M$.
First we show that $w_i\ge 2w_k$. By assumption, $p_i=v_i$. Hence
\begin{equation}
1\cdot w_i \ge (1+n_{i+1} + \hdots + n_k) \cdot 2\cdot w_k  \label{lemma_pacman_first_inequality}
\end{equation}
Similarly
\begin{equation}
w_k \ge (1+n_{k+1} + n_{k+2} + \hdots + n_{\beta+i}) \cdot w_{\beta+i} \ge n_{\beta+i} \cdot w_{\beta+i}  \label{lemma_pacman_second_inequality}
\end{equation}

Combining (\ref{lemma_pacman_first_inequality}) and (\ref{lemma_pacman_second_inequality}) we have
\begin{eqnarray*}\label{lemma_pacman_3}
n_i \cdot w_i - n_i \cdot w_k - n_{\beta+i}w_{\beta+i} & \geq &  n_i \cdot w_i - (n_i +1)\cdot w_k\\
&\ge& n_i \cdot w_i - (n_i +1)\cdot \frac{w_i}{2} \\
&\ge& 0
\end{eqnarray*}
as desired.
\end{proof}

\section*{Appendix 2: Incomplete Information.}
In this appendix, we introduce a restrictive incomplete information setting and show that the SPNE characterization obtained in Section \ref{sec:opt-cond} for the complete information setting also applies here.

Consider the setting in which the market participants can see who buys in period $t$, and know the distribution of values (and know their own value), but do not know exactly which consumer has which value.\footnote{Note this is equivalent to the participants knowing which consumer has which value, but not seeing who buys, only the total number of sales in each period.} Since we are interested in studying equilibria that satisfy the skimming property, regardless of the values of consumers who bought at period $t$, the duropolist off-path belief is that the $k$ consumers who bought in period $t$ are those with the $k$ highest valuations (among those remaining). We will show that the same conditions as in Section \ref{sec:opt-cond} characterizing the subgame perfect equilibria apply.

We first define the strategy of the duropolist and the consumers in any subgame under this incomplete information setting. Let $\mathcal{G}_{S}$ denote the subgame at period $t$ where the $|S|$ remaining consumers have valuations $w_1 \geq w_2 \hdots \geq w_{|S|}$. Due to the off-path belief (i.e., the belief that consumers follow the skimming property), the duropolist would behave as if it were in the market $\mathcal{G}'_{1}$ with $T-t+1$ periods in which the consumers are $v'_1 \geq v'_2 \geq \hdots, v'_{|S|}$ where $v'_i = v_{i+N -|S|}$. Observe that $v'_i \leq w_i$ for all $i \in [|S|]$. The monopolist strategy is to then announce the price $p_{\mathcal{G'}}^*(1,1)$ which is obtained by solving the recursion relationship (\ref{eq:characterization_prices}). The consumers strategy remains the same as in the complete information setting, i.e. each of them would buy if and only if the price is less than or equal to their threat price as calculated for $\mathcal{G'}$. We now prove the following result.

\begin{theorem} The strategies defined above constitute a SPNE in the incomplete information setting.\end{theorem}
\begin{proof}
We consider the subgame $\mathcal{G'} = \mathcal{G}_{S}$ (of the original game $\mathcal{G}$) that begins at period $t$ in which the remaining consumers consists of the set $S$. These consumers have valuations $w_1 \geq w_2 \hdots \geq w_{|S|}$. Due to the off-path belief (i.e., the belief that consumers follow the skimming property), the duropolist would behave as if it were in the market $\mathcal{G}'_{1}$ with $T-t+1$ periods in which the consumers are $v'_1 \geq v'_2 \geq \hdots, v'_{|S|}$ where $v'_i = v_{i+N -|S|}$. Observe that in this market $\mathcal{G}'_1$, consumer $i$'s real value is actually $w_i \geq v'_i$.

The announced price in $\mathcal{G}'_1$ would then be $p(1,t)= p(j^*(1, t), t+1)$. Suppose now that some consumer $x$ that was supposed to buy under the proposed equilibrium, i.e. $i \leq x \leq j^*(1, t)$ deviates and chooses not to buy at time $t$. The number of sales at period $t$ would then be $j^*(1, t)-1$, i.e., one less than the expected. The duropolist, who observes the total number of sales and assumes consumers follow the skimming property, would then behave as if the remaining subgame starting at $t+1$ is $\mathcal{G}'_{j^*(1, t)}$. This means that the announced price would be $p(j^*(1, t),t+1)$ and consumer $x$ would not have benefited from delaying the purchase by one period. One may, again, wonder whether consumer $x$ could benefit from delaying the purchase by more than one period. But this is not possible since the price at period $t+1$ in this subgame is $p(j^*(1, t),t+1) = p(j^*(j^*(1, t),t+1),t+2)$ , which means the price will remain constant over time, as long as the number of transactions is one less than the expected. By repeated use of this argument we conclude that, at equilibrium, no consumer would benefit from delaying its purchase.

Lastly, observe that no consumer can benefit from buying earlier. If a consumer deviates from the equilibrium path by buying earlier, she pays a price $p^*(1,t)$ when she could have bought in period $t'>t$ at price $p^*(k,t')$ for some $k\geq 1$. But since prices are non-increasing as a function of time along the proposed sales path (Lemma \ref{lem:EqPriceDec}), she cannot do any better.

So we conclude that we have a strategy profile which is an equilibrium in every subgame.
\end{proof}

Note that since the equilibrium path is the same in both our complete and incomplete information setting, all our results also apply to the incomplete information setting.

  \end{document}